\documentclass[11 pt]{amsart}
\usepackage{graphicx}

\usepackage{amsthm,amsmath,latexsym,amssymb,amsfonts,amscd,hyperref}
\usepackage{color,tikz,tikz-cd}
\usepackage[all]{xy}
\newtheorem{theorem}{Theorem}[section]

\newtheorem{proposition}[theorem]{Proposition}

\theoremstyle{definition}
\newtheorem{definition}[theorem]{Definition}
\newtheorem{example}[theorem]{Example}

\theoremstyle{remark}
\newtheorem{remark}[theorem]{Remark}

\numberwithin{equation}{section}

\begin{document}
\phantom{a}
\vspace{-1.5cm}
\title[Inner deformations of $A_\infty$-algebras]{Homotopy Cartan calculus and inner deformations of $A_\infty$-algebras} 

\author{ Alexey A. Sharapov }
\address{Physics Faculty, Tomsk State University, Lenin ave. 36, Tomsk 634050, Russia}
\email{sharapov@phys.tsu.ru}

\author{Evgeny D. Skvortsov}
\address{Service de Physique de l'Univers, Champs et Gravitation,  Universit\'e de Mons, 20 place du Parc, 7000 Mons, 
Belgium} 

\address{Lebedev Institute of Physics, 
Leninsky ave. 53, 119991 Moscow, Russia}
\email{evgeny.skvortsov@umons.ac.be}





\begin{abstract}
We consider inner deformations of families of $A_\infty$-algebras. With the help of noncommutative Cartan's calculus,  we prove the invariance of Hochschild (co)homology under inner deformations. The invariance also holds for cyclic cohomology classes that satisfy some additional conditions. Applications to dg-algebras and QFT problems are briefly discussed.
\end{abstract}

\maketitle
 
\section{Introduction}

 Strong homotopy algebras, mostly $L_\infty$, but $A_\infty$ too, have a long track record of applications in physics. Every gauge theory, for one, leads naturally to an $L_\infty$-algebra  within the Batalin--Vilkovisky formalism \cite{stasheff1998secret,jurvco2019algebras}. $A_\infty/L_\infty$-algebras are at the heart of open/closed string field theory \cite{Zwiebach:1992ie, Kajiura:2003ax, kajiura2006homotopy}. It is hardly possible to overestimate the organizing role of $L_\infty$-algebra language in the final solution of the deformation quantization problem \cite{Kontsevich:1997vb}.

 Our motivation for this work stems from the two recent incarnations of $L_\infty$-algebras: (i) there is a large class of integrable models \cite{Sharapov:2019vyd} relied on $L_\infty$-algebras; (ii) $L_\infty$-algebras occur as generalized global symmetries of certain conformal field  theories (CFT). In both cases the relevant $L_\infty$-algebras originate from $A_\infty$ ones through the symmetrization map (see Theorem \ref{Th10}, below). The study of observables and invariants in the models of (i) and (ii) calls for algebraic techniques we develop in the present paper. Let us elaborate on item (ii), which is, perhaps, the first occurrence of $L_\infty$ structures as a kind of {global} symmetry in physics. In quantum field theory (QFT), by an (infinitesimal) symmetry one usually means a Lie algebra action on a set of fields or operators which form, thereby, its module. One may regard these data as a very degenerate $L_\infty$-algebra with bilinear structure maps only. For example, this situation is realized for Chern--Simons matter vector models in the large-$N$ limit. What happens when one departs from this limit is that the Lie algebra and its module deform into an $L_\infty$-algebra, providing thus mathematical grounds \cite{Sharapov:2018kjz} to what was dubbed {\it slightly broken higher-spin symmetry} in \cite{Maldacena:2012sf}. One application of the present paper is  construction of various invariants, e.g. correlation functions, in such models \cite{Gerasimenko2022SlightlyBH}. 
 
 All $A_\infty$-algebras underlying the $L_\infty$-algebras of items (i) and (ii) share one important property: each of them is obtained by  {\it inner deformation} of a family of dg-algebras. The concept of inner deformation was introduced in our recent paper \cite{Sharapov:2018xxx}. In that paper, we show how to construct a minimal $A_\infty$-algebra starting from {\it any} one-parameter family of dg-algebras $A$. The resulting $A_\infty$-structure  extends the binary multiplication in $A$ by higher structure maps\footnote{Therefore it is not  just a minimal model of $A$.}. The construction is algorithmic and requires no extra data, hence the name.    From the physical viewpoint, adding a coherent set of higher products is equivalent to inclusion of interaction in a free gauge model. One may wonder what happens to the (co)homology of $A$, regarded as a graded associative algebra, upon inner deformations. In Sec. 7, we answer this question by showing that every inner deformation of an $A_\infty$-algebra extends to a deformation of its Hochschild (co)cycles (Theorem \ref{T61}).  In other words,  Hochschild (co)homology is invariant under inner deformations of families. Furthermore, the statement holds true for some classes of cyclic (co)homology, as is show in Sec. 8. 
 
 Our proofs make significant use of the noncommutative Cartan's calculus up to homotopy, as  developed in Refs. \cite{gel1989variant, Getzler93, nest1999cohomology, TamTsyn, dolgushev2011noncommutative}. Therefore, in the next Sec. 2-5, we recall relevant definitions and constructions. Guidance is also given, whenever necessary or helpful, on further reading. The original results of the paper are exposed in Sec. 6-8, where we consider deformation flows on families of $A_\infty$-algebras and their (co)homology. 
 
 \section{Gerstenhaber Algebras up to Homotopy}\label{s1}
 Throughout the paper  we work over a fixed ground field $k$ of characteristic zero. All unadorned tensor products $\otimes$ and $\mathrm{Hom}$'s are defined over $k$.  
When dealing with graded vector spaces and algebras, we tend to write down all sign factors explicitly. However, one can easily keep track of all signs using the Koszul rule: whenever two graded objects $a$ and $B$ are swapped, the factor of $(-1)^{|a||B|}$ appears.

 Given a pair of cochain complexes  $(C,\delta)$ and $(K,\partial)$, we say that a homomorphism $f : C^{\otimes n}\rightarrow K$ of degree $|f|\in \mathbb{Z}$ is a {\it $\delta\partial$-morphism} if 
$$
    \partial f=(-1)^{|f|}f\delta_n\,,
$$
where 
$$
     \delta_n(a_1\otimes \cdots\otimes a_n)=\sum_{k=1}^n (-1)^{|a_1|+\cdots+|a_{k-1}|}a_1\otimes\cdots\otimes \delta a_k\otimes \cdots\otimes a_n
$$
 is the standard differential in the $n$-th tensor power of the complex $C$.  Clearly, a $\delta\partial$-morphism of degree zero is just a cochain transformation.  
 
 A $\delta\partial$-morphism $f: C^{\otimes n}\rightarrow K$ is said to be {\it homotopic to zero} (in writing $f\simeq 0$) if there exists a homomorphism $g: C^{\otimes n}\rightarrow K$ of degree $|g|=|f|-1$ such that 
$$
     f=\partial g-(-1)^{|g|}g\delta_n\,.
$$
In the special case that $C=K$ and $\delta=\partial$ we speak about {\it $\delta$-morphisms}.

 
 \begin{definition}\label{D41} 
 A {\it Gerstenhaber algebra up to homotopy} is a cochain complex  $(G, \delta)$ composed of $k$-vector spaces and endowed with a pair of binary multiplication operations: a cup product $\cup$ of degree $1$ and a bracket $[\;\, ,\;]$ of degree $0$. In addition to $k$-bilinearity the multiplication operations are supposed to satisfy the following set of axioms: 
 \begin{itemize}
 \item [(hG1)] $\delta (a\cup b)=\delta a\cup b+(-1)^{|a|+1}a\cup \delta b$,
 \item [(hG2)] $\delta[a,b]=[\delta a, b]+(-1)^{|a|}[a,\delta b]$,
     \item [(hG3)] $(a\cup b)\cup c\simeq a\cup (b\cup c)\,,$
     \item [(hG4)] $a\cup b\simeq (-1)^{(|a|+1)(|b|+1)}b\cup a$\,,
     \item [(hG5)] $[a,b]\simeq -(-1)^{|a| |b|}[b,a]\,,$
     \item [(hG6)] $[a,[b,c]]\simeq [[a,b],c] +(-1)^{|a||b|}[b,[a,c]]\,,$
     \item [(hG7)] $[a,b\cup c]\simeq [a,b]\cup c+(-1)^{ |a|( |b|+1)} b\cup [a,c]$
 \end{itemize}
 for all homogeneous elements $a,b,c\in G$. 
 \end{definition}
 In what follows, we will abbreviate the verbose term `Gerstenhaber algebra up to homotopy' to {\it $hG$-algebra}.  
 
 The first two axioms imply that either multiplication operation defines a $\delta$-morphism from $G\otimes G$ to $G$. More precisely, the cup product gives the $\delta$-morphism $m(a,b)=(-1)^{|a|-1} a\cup b$, while the bracket defines a $\delta$-morphism as it is\footnote{Notice that the graded Leibniz rules  (hG1) and (hG7) are compatible with Koszul's sign convention if one thinks of $\cup$ as an object of degree $1$.}.  Then,  relations  (hG3,4) say that the cup product is associative and graded commutative up to homotopy. For example, unfolding the term `up to homotopy' for (hG3) we get 
 \begin{equation}\label{AuH}
   (a\cup b)\cup c-a\cup (b\cup c)
    \end{equation}
    $$
   =\delta A(a,b,c)+A(\delta a, b, c)+(-1)^{|a|}A(a, \delta b, c)+(-1)^{|a|+|b|}A(a,b,\delta c) 
$$
 for some homomorphism  $A: G\otimes G\otimes G \rightarrow G$ of degree $1$. Axioms (hG5,6) identify the bracket as a {Lie bracket} up to homotopy. Finally, relation (hG7) expresses compatibility between the two multiplication operations (the derivation property up to homotopy).

 \begin{example}
 Every cochain complex $(C,\delta)$ of $k$-vector spaces is an $hG$-algebra whose cup product and bracket are equal to zero. 
 \end{example}
 
 \begin{example}
 Every commutative dg-algebra $(A, \delta)$ gives rise to an $hG$-algebra on $G=A[1]$ w.r.t. the cup product $a\cup b= a\cdot b$ and the zero bracket\footnote{
 By definition, $A[m]$ is a graded vector space with $A[m]_n=A_{n+m}$.}. Similarly, one can think of each dg-Lie algebra $(L, \delta)$ as an $hG$-algebra with zero cup product. 
 \end{example}

\begin{example}\label{E44}
 If $\delta=0$, the above definition of an $hG$-algebra reduces to that of a {\it Gerstenhaber algebra}, see \cite{Gerst}. 
 This is defined by Axioms (hG3-7) with $\simeq$ replaced by the strict equality.  In the general case, both the cup product and the bracket pass through  the $\delta$-cohomology and induce a Gerstenhaber algebra structure on $H(G)$. 

A prototypical example of Gerstenhaber algebras is the algebra of poly\-vector fields $\mathcal{V}(M)$ on a smooth manifold $M$.  
If $\mathcal{V}^n(M)$ is the space of $n$-vector fields, then we set $G^{n-1}=\mathcal{V}^{n}(M)$. The role of cup product is played here  by the wedge product on polyvectors and the Lie bracket is given by the Schouten--Nijenhuis bracket. 
 
 \end{example}
 
In the examples above, all Axioms (hG1-7) were satisfied in the strong sense, and not just `up to homotopy'.  To come up with genuine examples of $hG$-algebras, we need to introduce the notions of  $A_\infty$- and $L_\infty$-algebras. These are defined as follows. 

Let $V=\bigoplus V_n$ be a $\mathbb{Z}$-graded vector space over $k$. Denote by $C_\bullet(V)$ the linear space of 
the reduced tensor algebra $$TV=\bigoplus_{n=1}^\infty V^{\otimes n}$$  and let $C^\bullet(V)$ denote the space of homomorphisms $\mathrm{Hom}(TV,V)$. The latter is known to carry the structure of a graded Lie algebra w.r.t. the {\it Gerstenhaber bracket}:
\begin{equation}\label{GB}
[A,B]=A\circ B-(-1)^{|A||B|}B\circ A\,,
\end{equation}
where 
\begin{equation}\label{comp}
(A\circ B)(a_1, a_2,\ldots, a_{m+n-1} )
\end{equation}
$$ =\sum_{k=0}^{n-1}(-1)^{|B|(|a_1|+\cdots +|a_k|)} A(a_1, \ldots, a_k, B(a_{k+1}, \ldots, a_{k+m}), \ldots,  a_{m+n-1})
$$
for all homogeneous $A\in C^n(V)$ and $B\in C^m(V)$. Here $|A|$ denotes the degree of $A$ as a homomorphism of graded vector spaces. It follows from the definition that $[A,A]=2A\circ A$ for every odd $A$.

\begin{definition}
 An  $A_{\infty}$-algebra structure on a graded vector space $V=\bigoplus V_n$ is given by a Maurer--Cartan element $m$
of the graded Lie algebra  $(C^\bullet(V), [\;\,, \; ])$, that is, 
\begin{equation}\label{MC}
m\circ m=0\,,\qquad |m|=1\,. 
\end{equation}
The pair $(V, m)$ is called an {\it $A_\infty$-algebra}.
\end{definition}

By definition, each  $A_\infty$-structure $m$ is given by an (infinite) sum 
\begin{equation}\label{mmm}
    m=m_1+m_2+m_3+\ldots 
\end{equation}
 of
multi-linear maps $m_n\in C^n(V)$ of degree $1$. It follows from (\ref{MC}) that the first structure map $m_1:V\rightarrow V$ squares to zero; and hence, it makes $V$ into a cochain complex with differential $\delta=m_1$. An $A_\infty$-algebra  is called {\it minimal} if $m_1=0$. The second structure map $m_2: V\otimes V\rightarrow V$  defines the multiplication operation $a\cup b=(-1)^{|a|-1}m_2(a,b)$, which is a cochain transformation of the complexes, i.e., $\delta $ differentiates the cup product by the Leibniz rule. The product obeys Rel. (\ref{AuH}) with $A(a,b,c)=(-1)^{|b|-1}m_3(a,b,c)$, meaning associativity up to homotopy.  Thus, we arrive at the following statement. 

\begin{proposition}\label{p46}
Let $(V,m)$ be an $A_\infty$-algebra with $m_2$ obeying the relation 
\begin{equation}\label{com}
    m_2(a,b)=-(-1)^{|a||b|}m_2(b,a)
\end{equation}
for all homogeneous $a,b\in V$. 
Then $V$ is a $hG$-algebra with 
$$
    \delta a=m_1(a)\,,\qquad a\cup b=(-1)^{|a|-1}m_2(a, b) \,, \quad \mbox{and}\quad [a,b]=0\,.
$$
\end{proposition}

The skew-symmetry condition (\ref{com}) admits a coherent extension to all higher  structure maps  leading to 
the notion of a {\it $C_\infty$-algebra}, see e.g. \cite{kadeishvili2008cohomology}, \cite{MARKL1992141}. Therefore, each $C_\infty$-algebra generates  an $hG$-algebra with trivial bracket.

Like $A_\infty$-algebras, $L_\infty$-algebras can be expressed in terms of $n$-ary operations subject to a coherent set of quadratic relations.
\begin{definition}\label{DL}
 An { $L_\infty$-algebra} structure on a graded vector space $V=\bigoplus V_n$ is given by a sequence of homomorphisms 
 $l_n:  V^{\otimes n}\rightarrow V$, $n=1,2,\ldots$, of degree $1$ that are graded symmetric,
 $$
 l_n(a_{1},\ldots, a_{i},a_{i+1},\ldots, a_{n})= (-1)^{|a_i| |a_{i+1}|}l_n(a_1,\ldots,a_{i+1},a_i,\ldots,a_n) \,,
 $$
 and satisfy the `generalized Jacobi identities':
 $$
\sum_{{}^{i+n=m}}     \sum_{{}^{\sigma\in \mathrm{Sh}(i,n)}} \!\!\!\!\!(-1)^{\kappa}l_{n+1}(l_i(a_{\sigma(1)},\ldots,a_{\sigma(i)}), a_{\sigma(i+1)}, \ldots, a_{\sigma(i+n)})=0
$$
 for all $m=1,2,\ldots$ Here $(-1)^\kappa$ is the Koszul sign resulting from permutations of homogeneous vectors $a_1,\ldots, a_n\in V$ and the second sum is over all $(i, n)$-shuffles, that is, permutations $\sigma\in S_{i+n}$ satisfying  $\sigma (k) < \sigma(k+1)$ for $k\neq i$.
\end{definition}

It follows from the definition that $l_1$ is a differential of degree $1$, i.e., $l_1^2=0$. The second structure map  $l_2$ endows $V[-1]$ with the skew-symmetric bracket 
$[a,b]=(-1)^{|a|+1}l_2(a,b)$, which is differentiated by $l_1$ and obeys the Jacobi identity up to homotopy. Explicitly, 
$$
\begin{array}{l}
    [[a,b], c]+(-1)^{(|a|+1)(|b|+|c|)}[[b, c],a]+(-1)^{(|c|+1)(|a|+|b|)}[[c,a],b]\\[3mm]
=\delta J(a,b,c)+J(\delta a, b, c)+(-1)^{|a|}J(a,\delta b, c)+(-1)^{|a|+|b|}J(a,b,\delta c)\,,
\end{array}
$$
where $\delta a=l_1(a)$ and $J(a,b,c)=(-1)^{|b|+1}l_3(a,b,c)$.
As with $A_\infty$-algebras, one can combine all the structure maps $\{l_n\}$  into a single  non-homogeneous map $l=l_1+l_2+l_3+\cdots\in C^\bullet(V)$ of degree $1$.
\begin{proposition}\label{p49}
Every $L_\infty$-algebra $(V, l)$ gives $G=V[-1]$ the structure of an $hG$-algebra with
$$
   \delta a=l_1(a)\,,\qquad [a,b]=(-1)^{|a|+1}l_2(a,b)\,,\quad\mbox{and}\quad a\cup b=0
$$
for all $a,b\in V$.
\end{proposition}
The next theorem states that full symmetrization of an $A_\infty$-algebra gives an $L_\infty$-algebra,  see e.g. \cite[Th. 3.1]{Lada1994StronglyHL}, \cite{Lada_commutators}.  

\begin{theorem}\label{Th10} Every $A_\infty$-algebra  $(V,m)$ induces an $L_\infty$-structure on the space $V$ with
\begin{equation}\label{ln}
    l_n(a_1,\ldots, a_n)=\sum_{\sigma\in S_n}(-1)^{\kappa}m_n(a_{\sigma(1)},\ldots, a_{\sigma(n)})\qquad \forall n>0\,,
\end{equation}
$(-1)^\kappa$ being the Koszul sign of the permutation $\sigma$.
\end{theorem}
For $n=2$, Eq. (\ref{ln}) takes the form 
\begin{equation}\label{LhG}
    l_2(a,b)=m_2(a,b)+(-1)^{|a||b|}m_2(b,a)\,.
\end{equation} Comparing this with Eq. (\ref{com}), we arrive at the following alternatives that are true for every $A_\infty$-algebra $(V, m)$ with $m_2\neq 0$:  
Either $V$ is an $hG$-algebra with a non-trivial cup product and the zero bracket (Prop. \ref{p46}) or $V[-1]$ is an $hG$-algebra with a 
non-trivial bracket and the cup product identically zero (Prop. \ref{p49}, Eq. (\ref{LhG})). 

Both the $A_\infty$- and $L_\infty$-algebras form categories, whose morphisms are most concisely defined in terms of coderivations, see Sec. \ref{S4}.

As the last example of $hG$-algebras we mention, without going into detail, a {\it homotopy Gerstenhaber algebra} on the Hochschild complex of an associative algebra
\cite{GV}, see Example \ref{E53} below. 

\section{Homotopy Cartan Calculus}

If $(K,\partial)$ is a cochain complex of $k$-vector spaces, then the space $\mathrm{End}(K)$ has the structure of a dg-algebra w.r.t. the composition of endomorphisms and the differential $\hat \partial$
defined by
$$
    \hat \partial A=[\partial, A] \qquad \forall A\in \mathrm{End}(K)\,.
$$

\begin{definition}  \label{D51}
 A {\it homotopy Cartan precalculus} over an $hG$-algebra $(G, \delta, \cup, [\;\,,\;])$ is a cochain complex of vector  spaces $(K,\partial)$ together with a pair of $k$-linear homomorphisms 
   \begin{equation}\label{iL}
       i: G\rightarrow \mathrm{End}(K)\,,\qquad L: G\rightarrow \mathrm{End}(K)
\end{equation}
   of degree $1$ and $0$, respectively, with the following properties:
\begin{itemize}
    \item[(hC1)] $ i_{\delta a}=-\hat\partial i_a$,
    \item[(hC2)] $L_{\delta a}=\hat\partial L_a$,
    \item[(hC3)] $i_ai_b\simeq i_{a\cup b}$,
    \item[(hC4)]  $[L_a,L_b]\simeq L_{[a,b]}$,
  \item[(hC5)] $[i_a, L_b]\simeq i_{[a,b]}$,
  \item[(hC6)]  $L_a i_b- (-1)^{ |a|} i_aL_b\simeq  L_{a\cup b}$
\end{itemize}
for all $a,b\in G$.     The first two conditions say that the maps (\ref{iL}) are $\delta\hat \partial$-morphisms and $L$ is a cochain transformation. 
 \end{definition}
The name `homotopy Cartan precalculus' comes from a special, albeit important, example where 
both the differentials $\delta$ and $\partial$ are zero (no homotopy). 
In that case, the remaining nontrivial Axioms (hC3-6) take the form of equalities satisfied by the operators of Lie derivative and 
interior product in the algebra of differential forms $\Lambda(M)$ on a smooth manifold $M$; in so doing, $G$ is identified with the Gerstenhaber algebra of polyvector fields $\mathcal{V}(M)$ (Example \ref{E44}), while $K^{-n}=\Lambda^n(M)$. One more important ingredient of the conventional Cartan's calculus is the exterior  differential $d$, which  enters the `Cartan magic formula'. 
Its natural extension to the  `up-to-homotopy' setting leads us to the following definition. 

\begin{definition}\label{D22}
 A homotopy Cartan precalculus is called {\it calculus}, if there exists an operator $d: K\rightarrow K$ of degree $-1$ such that 
 \begin{itemize}
 \item[(hC7)] $d\partial=-\partial d$\,,
     \item[(hC8)] $d^2\simeq 0$,
     \item[(hC9)] $d i_a+(-1)^{|a|} i_a d\simeq -L_a$.
 \end{itemize}
\end{definition}
Axiom (hC7) says that $d$ is a $\partial$-morphism of $K$.  Together Axioms (hC7-9) and (hC1) give one more familiar relation
\begin{equation}\label{dL}
    dL_a-(-1)^{|a|} L_a d\simeq 0\,.
\end{equation}
As a result, all the  maps $i$, $L$, and $d$ pass through $\delta$- and $\hat\partial$-cohomology, inducing the conventional  Cartan's calculus relations in $H(\mathrm{End}(K))$ for the Gerstenhaber algebra $H(G)$.  

In the sequel, we abbreviate `homotopy Cartan (pre)calculus' to $hC$-(pre)calculus. 

\begin{example}\label{E53}
Let $A$ be a graded associative algebra and  let $V=A[1]$.  The spaces $C_\bullet (V)$ and $C^\bullet(V)$ are 
identified, respectively, with the  spaces of Hochschild chains and cochains of the algebra $A$ with coefficients in itself. The corresponding Hochschild differentials are given by 
 $$
     (\delta D)(a_0,\ldots, a_n)
$$
$$
=(-1)^{|D|+\bar a_0+\cdots +\bar a_{n-1}}D(a_{0}, \ldots, {a_{n-1}})a_n+(-1)^{(|D|+1)\bar a_0} a_0D(a_1,\ldots, a_n)
$$
$$
-\sum_{i=0}^{n-1} (-1)^{|D|+\bar a_0+\cdots +\bar a_i}D(a_0,\ldots, a_{i-1}, a_{i}a_{i+1},a_{i+2},\ldots, a_n)
$$
for all $D\in C^{n}(V)$ and 
$$
    \partial (a_0\otimes a_1\otimes \cdots \otimes a_n)
$$
$$
=(-1)^{\bar a_n(\bar a_0+\cdots +\bar a_n+1)}a_na_0\otimes a_1\otimes \cdots\otimes a_{n-1}
$$
$$
+\sum_{i=0}^{n-1} (-1)^{\bar a_0+\cdots+\bar a_i} a_0\otimes \cdots\otimes a_{i-1}\otimes a_ia_{i+1}\otimes a_{i+2}\otimes \cdots\otimes  a_n\,.
$$
Here $\bar a=|a|-1$ is the degree of $a\in A$ as an element of $V=A[1]$. Notice that $\delta D=[m_2, D]$, where $m_2(a,b)=(-1)^{\bar a} ab$. 

As was first observed in \cite{Gerst}, the  cochain complex $C^\bullet(V)$ carries the structure of an $hG$-algebra for the Gerstenhaber bracket (\ref{GB}) and the cup product 
 $$
     (D\cup E)(a_1,\ldots, a_{k+l})
 $$
 $$
 =(-1)^{(|E|+1)(\bar a_1+\cdots+\bar a_k)}D(a_1,\ldots, a_k)E(a_{k+1},\ldots, a_{k+l})
 $$
 for all $D\in C^k(V)$ and $E\in C^l(V)$. 

Natural action of Hochschild cochains on chains gives rise to an $hC$-precalculus with 
 \begin{equation}\label{iD}
    i_D (a_0\otimes  \cdots \otimes a_n) 
 \end{equation}
 $$
 = (-1)^{(\bar a_0+\cdots+\bar a_n)(\bar a_{n-k+1}+\cdots+ \bar a_n)+|D|-1} D(a_{n-k+1},\ldots,a_n ) a_0\otimes a_{1}\otimes \cdots\otimes  a_{n-k}
 $$
 and 
 \begin{equation}\label{LE}
     L_E (a_0\otimes \cdots \otimes a_n)
 \end{equation}
 $$
 =\sum_{i=0}^{n-k} (-1)^{|E|(\bar a_0+\cdots+\bar a_i)}a_0\otimes \cdots\otimes a_i\otimes E(a_{i+1},\ldots, a_{i+k})\otimes a_{i+k+1}\otimes \cdots\otimes a_n
 $$
 $$
 +\sum^{n}_{i=n-k} (-1)^{(\bar a_0+\cdots +\bar a_n+1)(\bar a_{i+1}+\cdots + \bar a_n)} E(a_{i+1},\ldots, a_{n}, a_0, \ldots, a_{i+k-n-1})
 $$
 $$
 \otimes a_{i+k-n}\otimes \cdots \otimes a_i
 $$
 for all $D,E\in C^k(V)$. Direct  verification shows that the operators (\ref{iD}) and (\ref{LE}) obey Axioms (hC3,4)  with strict equalities.
 
 
 For a unital algebra $A$, one can also define the normalized Hochschild complexes  $\bar C^\bullet(V)$ and $\bar C_\bullet(V)$ together with  the  Connes--Rinehart boundary operator 
 \begin{equation}\label{B}
 \begin{array}{c}
     B(a_0\otimes \cdots \otimes a_n)\\[3mm]
     \displaystyle =\sum_{i=0}^n (-1)^{(\bar a_0+\cdots+\bar a_i)(\bar a_{i+1}+\cdots+\bar a_n)}e\otimes a_{i+1}\otimes \cdots\otimes a_n
     \otimes a_0\otimes \cdots \otimes a_i\,.
     \end{array}
 \end{equation}
Here $e$ is the unit of $A$, $a_0\in A$, and  $a_1,\ldots,a_n\in \bar A=A/ke$.  This makes the above $hC$-pre\-calculus into an $hC$-calculus with differential $d=B$. 
Notice that $d^2=0$ and the l.h.s. of Rel. (\ref{dL}) vanishes in the usual sense.  For more detail, see  \cite{gel1989variant}, \cite{Cuntz}, \cite[Sec. 3]{DomKow}. 
 
\end{example}

\section{$hG$-algebra structure and $hC$-calculus on the Hochschild complex of an $A_\infty$-algebra}\label{S3}
In this section, we extend Example \ref{E53} to an arbitrary $A_\infty$-algebra $(V,m)$. First of all, the Gerstenhaber bracket (\ref{GB}) makes $C^\bullet(V)$ into a dg-Lie algebra with the differential  given by the adjoint action of $m$, i.e., $\delta A=[m,A]$. 
To define a cup product on $C^\bullet(V)$, we need the construction of  {\it braces} \cite{Kadeishvili98}, \cite{Getzler93}, \cite{GV}. 

\begin{definition}\label{D61}
 For homogeneous elements  $A, A_1,\ldots, A_m\in C^\bullet(V)$ and  $a_1,\ldots, a_n\in V$, define the braces $A\{A_1,\ldots, A_m\}\in C^\bullet(V)$ by the formula
$$
\begin{array}{c}
  A\{A_1,\ldots, A_m\}(a_1,\ldots, a_n) \\[2mm] \displaystyle
    =\sum_{0\leq k_1\leq \cdots\leq k_m\leq n} (-1)^\kappa A(a_1,\ldots, a_{k_1}, A_1(a_{k_1+1}, \ldots), \\[2mm]  \hspace{3cm}\ldots,a_{k_m}, A_{m}(a_{k_m+1},\ldots),\ldots,a_n)\,,
\end{array}
$$
where $\kappa =\sum_{i=1}^m|A_i|\sum_{j=1}^{k_i}|a_j|$. It is assumed that $A\{\varnothing\}=A$. 
\end{definition}
It follows from the definition that 
$A\{A_1\}=A\circ A_1$ and one may think of the braces as a generalization of the composition product (\ref{comp}).
The braces obey the {\it higher pre-Jacobi identities} \cite{GV}:
\begin{equation}\label{pre-J}
\begin{array}{c}
    A\{A_1,\ldots, A_m\}\{B_1,\ldots, B_n\}\\[2mm]
\displaystyle =\sum_{\mbox{\small $AB$-shuffles}} (-1)^\kappa A\{B_1,\ldots, B_{k_1},A_1\{B_{k_1+1}, \ldots\},\\[2mm]
\hspace{2,5cm}\ldots, B_{k_m}, A_m\{B_{k_m+1},\ldots\},\ldots, B_n\}\,,
\end{array}
\end{equation}
where $\kappa=\sum_{i=1}^m|A_i|\sum_{j=1}^{k_i}|B_j|$.
Here summation is over all shuffles of the $A$'s and $B$'s (i.e., 
the order of elements in either group is preserved under permutations) and
the case of empty braces  $A_k\{\varnothing\}$ is not excluded.

In \cite{Getzler93}, Getzler shows that every $A_\infty$-structure $m$ on $V$ induces an $A_\infty$-structure $M$ on $C^\bullet(V)$. The latter is given by
\begin{equation}\label{M}
\begin{array}{l}
   M_0(\varnothing)=0\,,\\[2mm]
   M_1(A)=m\circ A-(-1)^{|A|}A\circ m \,,\\[3mm]
M_k(A_1,\ldots,A_k)=m\{A_1,\ldots,A_k\}\,,\qquad k>1\,.
   \end{array}
\end{equation}
Indeed, applying  (\ref{pre-J}) one finds 
\begin{equation}\label{BI}
    (M\circ M)(A_1,\ldots, A_n)=(m\circ m) \{A_1,\ldots, A_n\}=0\,.
\end{equation}
Notice that 
\begin{equation}\label{M1}
    \delta A= [m,A]=M_1(A)\,.
\end{equation}

\begin{proposition} \label{p32}
The dg-Lie algebra $(C^\bullet(V), \delta, [\;\,,\;])$ is an $hG$-algebra 
for the cup product 
\begin{equation}\label{CP}
A\cup B=(-1)^{|A|-1}M_2(A,B)\,.
\end{equation}
\end{proposition}

\begin{proof}
By the definition of an $A_\infty$-structure,
\begin{equation}
    M_1\circ M_1=0\,,\qquad [M_1, M_2]=0\,,\qquad M_2\circ M_2=-[M_1,M_3]\,.
\end{equation}
The first relation identifies the operator (\ref{M1}) as differential, i.e., $\delta^2=0$. Then the  second equality amounts to  Axiom  (hG1) of Definition \ref{D41}. 
Axiom (hG3) -- associativity up to homotopy -- is equivalent to the third relation. The commutativity up to homotopy (hG4) follows from the identity
$$
\begin{array}{c}
    M_2(A,B)+(-1)^{|A||B|}M_2(B,A)\\[2mm]
= -M_1(A\circ B)+M_1(A)\circ B+(-1)^{|A|} A\circ M_1(B)\,,
\end{array}
$$
which is a consequence of the pre-Jacobi identities (\ref{pre-J}).
The same pre-Jacobi identities lead to the equality 
$$
\begin{array}{c}
[A,M_2(B, C)]-(-1)^{|A|}M_2([A,B], C) -(-1)^{|A|(|B|+1)}M_2(B, [A,C])\\[2mm]
=(-1)^{|A|} \Big( M_1(A\{B,C\}) - M_1(A)\{B,C\}\\[2mm]
-(-1)^{|A|}A\{M_1(B),C\}-(-1)^{|A|+|B|}A\{B,M_1(C)\}\Big)\,, 
\end{array}
$$
which is equivalent to Axiom (hG7).
\end{proof}

Now we would like to equip the $hG$-algebra above with an $hC$-calculus following the pattern of Example \ref{E53}. To this end, let us consider the graded associative algebra  of endomorphisms of $C_\bullet(V)$ as an $A_\infty$-algebra $\mathrm{End}(C_\bullet(V))[1]$.    

\begin{theorem}[\cite{Getzler93}]
There is a morphism from the $A_\infty$-algebra $(C^\bullet(V), M)$ to $\mathrm{End}(C_\bullet(V))[1]$. 
\end{theorem}

By definition, the morphism in question is given by a map $$I: TC^\bullet(V)\rightarrow \mathrm{End}(C_\bullet(V))[1]$$ obeying the sequence of relations
$$
    \sum_{k=0}^n (-1)^{|A_1|+\cdots+|A_k|}I(A_1,\ldots, A_k) I(A_{k+1},\ldots, A_{n})=-(I\circ M)(A_1,\ldots, A_{n})\,.
$$
For $n=0,1,2$ we get 
$$
    I(\varnothing)I(\varnothing)=0\,, \qquad [I(\varnothing), I(A)]=-I(M_1(A))\,,
$$
$$
 [I(\varnothing), I(A,B)]+(-1)^{|A|}I(A)I(B)
$$
$$
=-I(M_1(A),B)-(-1)^{|A|}I(A, M_1(B))-I(M_2(A,B))\,.
$$
Hence we will satisfy Axiom (hC3) of Definition \ref{D51}, if set
$$
   \partial=I(\varnothing)\,, \qquad i_A=I(A)\,. 
$$
Here is an explicit expression for the morphism $I$:
\begin{equation}\label{IO}
     I(\varnothing)(a_0,\ldots, a_n)
\end{equation}
$$
=\sum_{l=1}^\infty \sum_{j=n-l+1}^n (-1)^{(|a_0|+\cdots+|a_j|)(|a_{j+1}|+\cdots+|a_n|)}(m_l(a_{j+1}, \ldots, a_n,a_0,\ldots),\ldots, a_{j})
$$
$$
+\sum_{l=1}^\infty \sum_{j=0}^{n-l}(-1)^{|a_0|+\cdots+ |a_j|} (a_0,\ldots, m_l(a_{j+1},\ldots, a_{j+l}),\ldots, a_n)
$$
($a_0$ is inside $m_l$ in the first double sum) and

\begin{equation}\label{I}
    I(A_1,\ldots, A_k)(a_0,\ldots, a_n)
\end{equation}
$$
=\sum_{l=k+1}^\infty \sum_{j_0,\ldots,j_k} (-1)^{(|a_0|+\cdots +|a_{j_0}|)(|a_{j_0+1}|+\cdots+|a_n|)+\sum_{i=1}^k |A_i|(|a_0|+\cdots +|a_{j_i}|)} 
$$
$$
\times (m_l(a_{j_0+1},\ldots, A_1(a_{j_1+1},\ldots),\ldots, A_{k}(a_{j_k+1},\ldots), \ldots,a_n,a_0,\ldots),\ldots, a_{j_0})
$$
for $k>0$. The sum is taken over all possible values of $j_0,\ldots,j_k$ such that $a_0$ is to the right of $A_k$ and both are inside $m_l$. 
Hereinafter we denote by $(a_0,\ldots, a_n)$ the chain $a_0\otimes \cdots\otimes a_n$ of $C_n(V)$. 

In order to define the Lie derivative operation, consider the map 
$$
    L : TC^\bullet(V)\rightarrow \mathrm{End}(C_\bullet(V))
$$
given by 
\begin{equation}\label{LAA}
    L(A_1,\ldots, A_k)(a_0,\ldots, a_n)=\sum_{j_1,\ldots,j_k}(-1)^{\sum_{i=1}^k |A_i|(|a_0|+\cdots +|a_{j_i}|)} 
\end{equation}
$$
\times (a_0,\ldots, A_1(a_{j_1+1},\ldots),\ldots, A_k(a_{j_k+1},\ldots),\ldots,a_n)
$$
$$
+\sum_{i=1}^k\sum_{j_1,\ldots,j_k}(-1)^{(|A_2|+\cdots+|A_i|)(|A_{i+1}|+\cdots +|A_k|)+(|a_0|+\cdots+|a_{j_1}|)(|a_{j_1+1}|+\cdots+|a_n|)} 
$$
$$
\times (-1)^{\sum_{l=2}^i|A_l|(|a_0|+\cdots+|a_{j_l}|+|a_{j_1+1}|+\cdots +|a_n|)+\sum_{l=i+1}^k|A_l|(|a_{j_1+1}|+\cdots+|a_{j_l}|)}
$$
$$
\times (A_1(a_{j_1+1},\ldots, A_{i+1}(a_{j_1+1}, \ldots), \ldots, A_k(a_{j_k+1},\ldots),\ldots,a_0,\ldots),
$$
$$
 \ldots, A_2(\ldots),\ldots, A_i(\ldots),\ldots, a_{j_1})\,.
 $$
 Summation is taken over all possible values $j_1,\ldots, j_k$ subject to the following condition: in the first single sum $a_0$ is to the  left of $A_1$, while in the second double sum $a_0$ is inside $A_1$ and to the right of $A_k$. For a single cochain $A\in C^p(V)$, Eq. (\ref{LAA}) simplifies to 
 $$
    L(A)(a_0,\ldots, a_n)=
$$
$$
\sum_{j=0}^{n-p}(-1)^{ |A|(|a_0|+\cdots +|a_{j}|)} (a_0,\ldots, A(a_{j+1},\ldots, a_{j+p}),\ldots,a_n)
$$
$$
+\!\!\!\!\!\!\sum_{j=n-p+1}^n\!\!\!\!\!(-1)^{(|a_0|+\cdots+|a_j|)(|a_{j+1}|+\cdots +|a_n|)}(A(a_{j+1},\ldots,a_0, \ldots, a_{j+p-n-1}),\ldots, a_j).
$$
Comparing the last expression with (\ref{IO}), we conclude that $\partial= L(m)$. 

In \cite{Getzler93}, Getzler proved the following identity for $I$ and $L$:
\begin{equation}
    L\circ M=IL-LI
\end{equation}
Explicitly,
$$
   \sum_{k=1}^\infty \sum_{j=0} ^{n-k}(-1)^{|A_1|+\cdots +|A_j|}I(A_1,\ldots, A_j,M_k(A_{j+1},\ldots A_{j+k}),\ldots, A_n)
$$
$$
= \sum_{j=0}^{n-1}I(A_1,\ldots, A_j) L(A_{j+1},\ldots, A_n)
$$
$$
-\sum_{j=1}^n(-1)^{|A_1|+\cdots+|A_j|}L(A_1,\ldots, A_j)I(A_{j+1},\ldots, A_n)
$$
for $n=1,2,\ldots$ In particular, 
$$
    I(\varnothing)L(A)-(-1)^{|A|}L(A)I(\varnothing)=L(M_1(A))
$$
and 
$$
\begin{array}{l}
 I(A){L}(B)-(-1)^{|A|}L(A)I(B)+[I(\varnothing), L(A,B)]\\[3mm]
 =L(M_1(A), B)+(-1)^{|A|}L(A,M_1(B))+{L}(M_2(A,B))\,.
 \end{array}
 $$
 Therefore, we can satisfy Axioms (hC2) and (hC6) by setting $
     L_A=L(A)$.
 The remaining two Axioms (hC4) and (hC5) can now be verified  directly. As shown in \cite[Lem. 2.3]{Getzler93}, the assignment  $A \mapsto L_A$ defines a hom\-omor\-phism of graded Lie algebras, that is,
 \begin{equation}\label{Lrep}
     [L_A,L_B]=L_{[A,B]}\,.
 \end{equation}
 Axiom (hC5)
 can be written as 
 $$
     (-1)^{|A|}[L_A,i_B]=i_{[A,B]} +\hat\partial T(A,B)-T (\delta A,B)-(-1)^{|A|}T(A,\delta B)\,,
 $$
or equivalently, 
\begin{equation}
    i_BL_A-(-1)^{|A|(|B|+1)}L_Ai_B=i_{[B,A]}
\end{equation}
$$
+(-1)^{|A||B|}\Big( T(\delta A,B)+(-1)^{|A|}T(A,\delta B)-\partial T(A,B)+(-1)^{|A|+|B|}T(A,B)\partial \Big),
$$
where the Gel'fand--Daletskii--Tsygan homotopy $T(A,B)$ is given by 
\begin{equation}\label{TAB}
    T(A,B)(a_0,\ldots, a_n)
\end{equation}
$$
=\sum_{j,i} (-1)^{\kappa}(A(a_{j+1},\ldots, B(a_{i+1},\ldots, a_{q+1}), \ldots, a_0,\ldots), \ldots, a_j)\,.
$$
Here $(-1)^\kappa$ is the standard Koszul sign and the sum is taken over all $i$ and $j$ such that $a_0$ is inside $A$ but to the right  of $B$. (As above, all permutations preserve the cyclic order of $a$'s.)
In a slightly different notation an explicit  expression for $T(A,B)$ was given in \cite{Getzler93} and \cite[Ex. 3.13]{DomKow}; in the former paper it is denoted by $\rho$.

The above operations of Lie derivative $L_A$ and contraction $i_A$ can further be  supplemented  with a differential $d$ to give an example of $hC$-calculus.
To this end, one adds a new element $e$ of degree $-1$ and defines the augmented vector space $\tilde V=V\oplus ke$. The $A_\infty$-structure $m$ extends to $\tilde V$ by setting 
$$
m_2(e,a)=(-1)^{|a|}m_2(a,e)=a\,, \qquad m_2(e,e)=e\,,\qquad m_k(\ldots,e,\ldots)=0 
$$
for all $k\neq 2$ and $a\in V$. The space $\bar C_n(V)$ of normalized Hochschild $n$-chains is spanned by the tensor products $a_0\otimes a_1\otimes\cdots\otimes a_n$, where $a_0\in \tilde V$ and $a_1,\ldots,a_n\in V$.
The differential $d$, being independent of a particular $A_\infty$-structure on $V$, is given  by the Connes--Rinehart operator (\ref{B}). Since the symbol $d$ is usually overloaded, we will denote this differential by $B$, which is more standard in the case of $\bar C_\bullet(V)$.  
Axioms (hC8,9) of Definition \ref{D22} take now the form
$$
    {B}^2=0\,,\qquad [B, i_A]=-L_A -\hat \partial S_A -S_{\delta A}\,,
$$
where 
$$
    S_A(a_0,\ldots,a_n)=\sum_{j\leq i} (-1)^\kappa (e, a_{j+1},\ldots, A(a_{i+1},\ldots, a_{i+p}),\ldots, a_0,\ldots, a_j)
$$
and the Koszul sign is given by
$$
\kappa={(|a_0|+\cdots+|a_j|)(|a_{j+1}|+\cdots+|a_{n}|)+|A|(|a_{j+1}|+\cdots +|a_i|)} \,.
$$
The sum is taken over all cyclic permutations of $a$'s such that $a_0$ appears to the right of $A\in C^p(V)$. 
\section{Coderivations vs. Lie derivatives}\label{S4}

Recall that, in addition to the
associative algebra structure, the space $C_\bullet (V)=TV$ carries the
structure of a coassociative  coalgebra with respect to the
(reduced) coproduct
$
\Delta: TV\rightarrow TV\otimes TV\,,
$
$$
\Delta (a_1\otimes \cdots\otimes a_n)=\sum_{i=1}^{n-1}(a_1\otimes\cdots\otimes a_i)\otimes
(a_{i+1}\otimes\cdots\otimes a_n)\,.$$
Coassociativity is expressed by the relation $(1\otimes
\Delta)\Delta=(\Delta\otimes 1)\Delta$.

A linear map $D: TV\rightarrow TV$ is  a {\it
coderivation}, if it obeys the co-Leibniz rule $$\Delta D=(D\otimes
1+1\otimes D)\Delta\,.$$ The space of coderivations is known to be
isomorphic to the space $C^\bullet(V)=\mathrm{Hom}(TV,V)$, so
that  any homomorphism $A: TV\rightarrow V$ induces a coderivation
$D_A: TV\rightarrow TV$ and vice versa:  if $A\in
C^m(V)$, then
$$
\begin{array}{rl}
D_A(a_1\otimes\cdots\otimes a_n)=\displaystyle \sum_{i=1}^{n-m+1}&(-1)^{|A|(|a_1|+\cdots+ |a_{i-1}|)}a_1\otimes\cdots\otimes a_{i-1}\\[4mm]
&\otimes A(a_i\otimes \cdots \otimes a_{i+m-1})\otimes
a_{i+m}\otimes \cdots\otimes a_n
\end{array}
$$
for $n\geq m$ and zero otherwise. This allows
one to interpret the Gerstenhaber bracket (\ref{GB}) as the
commutator of two coderivations. More precisely, the assignment $A\mapsto D_A$ defines a homomorphism from the Lie algebra  $C^\bullet (V)$ for
the Gerstenhaber bracket (\ref{GB}) to the Lie algebra of coderivations:
\begin{equation}\label{Drep}
[D_A, D_B]=D_{[A,B]}\qquad \forall A, B\in C^\bullet (V)\,.    
\end{equation}

Thus,  we have two representations of the graded Lie algebra $(C^\bullet(V), [\;\,,\;])$ in the space $C_{\bullet}(V)$: 
by Lie derivatives (\ref{Lrep}) and by coderivations (\ref{Drep}).    
Neither representation is irreducible. To describe invariant subspaces, we introduce the generator $\lambda$ of cyclic permutations,
$$
    \lambda (a_0\otimes a_1\otimes \cdots\otimes a_n)=(-1)^{|a_{0}|(|a_1|+\cdots +|a_{n}|)} (a_1\otimes\cdots\otimes a_{n}\otimes a_0)\,,
$$
and let $N=1+\lambda+\lambda^2+\cdots+\lambda^{n}$ denote the associated norm map. 
Since $\lambda^{n+1}=1$, we have
$
    N(1-\lambda)=(1-\lambda)N=0
$ on $C_n(V)$.
The operators $N$ and $1-\lambda$ intertwine the $L$- and $D$-representations of the Lie algebra $C^\bullet(V)$, that is, 
\begin{equation}\label{N}
    NL_A=D_AN\,,\qquad L_A(1-\lambda)=(1-\lambda)D_A\qquad \forall A\in C^\bullet(V)\,.
\end{equation}
It follows from these relations that the subspace $\mathrm{Im} (1-\lambda)=\ker N \subset C_\bullet(V)$ is invariant under the action of  $L_A$, while $\mathrm{Im} N= \ker (1-\lambda)\subset C_\bullet(V)$ is an invariant subspace for $D_A$. 

If $(V,m)$ is an $A_\infty$-algebra, then the coderivation $D_m$ squares to zero and makes the coassociative coalgebra $(TV, \Delta)$ into a codifferential coalgebra $(TV, \Delta, D_m )$. This allows one to define morphisms of $A_\infty$-algebras $(V,m)$ and $(V',m')$ as homomorphisms $h: TV\rightarrow TV'$ of the corresponding codifferential coalgebras. 

The same language of coderivations applies to $L_\infty$-algebras and their morphisms \cite{StasheffLada}. This time one starts with the symmetric algebra $SV$ of a graded space $V$. This is a subcoalgebra $SV\subset TV$ spanned by symmetric tensors. Upon restriction to $SV$ one gets the coproduct 
$$
\Delta (a_1\vee \cdots\vee a_n)=\sum_{i=1}^{n-1}\sum_{{}^{\sigma\in\mathrm{Sh}(i,n-i)}}\!\!\!\!(-1)^\kappa(a_{\sigma(1)}\vee\cdots\vee a_{\sigma(i)})\otimes
(a_{\sigma(i+1)}\vee\cdots\vee a_{\sigma(n)}).$$
 An $L_\infty$-structure on $V$ can now be defined as a coderivation of $(SV, \Delta)$ that has degree one and squares to zero.  Again, it is easy to see that each coderivation $D$ of $(SV,\Delta)$ (and hence, $L_\infty$-structure) defines and is defined by some element of $\mathrm{Hom}(SV,V)$. Explicitly, 
$$
    D_A(a_1\vee \cdots\vee a_n)=\!\!\sum_{{}^{\sigma\in \mathrm{Sh}(i,n-i)}}\!\! (-1)^\kappa A(a_{\sigma(1)}\vee\cdots\vee a_{\sigma(i)})\vee a_{\sigma(i+1)}\cdots\vee a_{\sigma(n)}
$$
for all $A\in \mathrm{Hom}(S^iV, V) $. The generalized Jacobi identities of Definition \ref{DL} are equivalent  to the condition $D^2_l=0$. Thus, there is a one-to-one correspondence between $L_\infty$-structures on $V$ and codifferentials on the coalgebra $(SV,\Delta)$. This makes possible to define the morphisms of $L_\infty$-algebras as the homomorphisms of the corresponding codifferential coalgebras $(SV, \Delta, D_l)$.   Theorem \ref{Th10}  provides then a functor from the category of $A_\infty$-algebras to the category of $L_\infty$-algebras, the {\it symmetrization map}.

For every $A_\infty$-algebra $(V,m)$ the operator of Lie derivative $\partial=L_m$  makes the space $C_\bullet(V)$ into a complex, which we denote by $C_\bullet(V,m)$. The  complex $C_\bullet(V, m)$ is  known as the Hochschild chain complex of the $A_\infty$-algebra $(V,m)$ with coefficients in  itself. We let denote its homology groups by $HH_\bullet(V,m)$. For general definitions of 
$A_\infty$-(bi)modules and $A_\infty$-(co)homology we refer the reader to \cite{Kadeishvili1980ONTH}, \cite[Sec. 3]{Getzler1990A}, \cite{MARKL1992141}. 
It follows from the second relation in (\ref{N}) that the subspace $\mathrm{Im}(1-\lambda)\subset C_\bullet(V)$ is a subcomplex of $C_\bullet(V, m)$. One may regard the corresponding quotient complex 
$$
    C_\bullet^{\lambda}(V,m)=C_{\bullet}(V,m)/\mathrm{Im}(1-\lambda)
$$
as an $A_\infty$ generalization of the cyclic chain complex of an associative algebra \cite{Loday}. Its homology will be denoted by $HC_\bullet^\lambda(V,m)$.
Passing to the dual space $C_\bullet (V)^\ast$,  one can also define the corresponding cochain complexes $C^\bullet(V,m)$ and $C^\bullet_\lambda (V, m)$ of an $A_\infty$-algebra $(V,m)$.  The latter  consists of cochains $S: TV\rightarrow k$ that obey the cyclicity condition $S(1-\lambda)=0$. We denote the cohomology groups of these complexes by $HH^\bullet (V,m)$ and $HC^\bullet_\lambda(V,m)$, respectively.

\begin{remark}
A word about terminology. In the above notation for the (co)chain complexes and their (co)homology groups,  we use the label $\bullet$ just to distinguish between  the spaces $TV$ and $\mathrm{Hom}(TV,k)$, i.e., the Hochschild chains and cochains. The usual simplicial degree associated with the number of factors in tensor products $a_0\otimes\cdots\otimes a_n$ is not essential for our considerations as the differential $\partial$ is not generally homogeneous relative to it. 
Since $|\partial|=1$, all the aforementioned complexes are {\it cochain} complexes with respect to the $\mathbb{Z}$-degree, though this is not indicated in  
their notation explicitly.  However, as in the case of associative algebras, we still call the elements of $C_\bullet(V,m)$ chains (cycles) and the elements of $C^\bullet(V,m)$ cochains (cocycles). 
\end{remark}

In much the same way one can  define the (co)homology theory for $L_\infty$-algebras with various coefficients. In particular, the complex $(SV, D_l)$ above computes the homology of an $L_\infty$-algebra  with trivial coefficients. For further references, we also introduce the corresponding cochain complex $(SV)^\ast$ and denote its cohomology groups by $H^\bullet(V,l)$.  A detailed discussion of $L_\infty$-modules and $L_\infty$-(co)homology can be found in 
\cite{Lada1994StronglyHL}, \cite{Reinhold}.

\section{Families of $A_\infty$-algebras and their inner deformations}

Let $(V,m)$ be formal $n$-parameter deformation of an $A_\infty$-algebra $(V,m_0)$, that is, an $A_\infty$-structure $m$ is given by a Maurer--Cartan element $m\in C^\bullet(V)[[t_1,\ldots, t_n]]$ such that $m|_{t=0}=m_0$. For more generality we  allow the deformation parameters $t_i$ to carry   {\it even} $\mathbb{Z}$-degrees. We will refer to $(V,m)$ as a {\it family of $A_\infty$-algebras}. 
According to Proposition \ref{p32}, the family $(V,m)$ gives rise to the family of $hG$-algebras ${G}_t=\big(C^\bullet(V)[[t_1,\ldots,t_n]], \delta, [\;\,,\;], \cup \big)$. 
Let us denote 
$$
m_{(i_1i_2\cdots i_k)}=\frac{\partial^k m}{\partial t_{i_1}\partial t_{i_2}\cdots \partial t_{i_k}}\in C^\bullet(V)[[t_1,\ldots,t_n]]\,.
$$
It is clear that $|m_{(i_1i_2\cdots i_k)}|={1 -|t_{i_1}|-\ldots-|t_{i_k}|}$. We will also use this notation for partial derivatives of other cochains. Differentiating the defining relation $m\circ m=0$ w.r.t. the parameters, we get
\begin{equation}\label{mm'}
[m, m_{(i)}]=\delta m_{(i)}=0\,.
\end{equation}
Hence, $m_{(i)}$ is a $\delta$-cocycle representing a cohomology class of $H^{1-|t_i|}(G_t)$. As discussed in Sec. 2, the bracket and  the cup product pass through the cohomology and make the graded space $H^\bullet(G_t)$ into a Gerstenhaber algebra. Let $\mathcal{I}=\bigoplus \mathcal{I}^p$ denote the subalgebra of $G_t$ 
generated by the cocycles $m_{(i)}$. Passing to the $\delta$-cohomology, we also define the subalgebra $H^\bullet(\mathcal{I})$ of the Gerstenhaber algebra $H^\bullet(G_t)$.
\begin{proposition}
The Gerstenhaber bracket on $G_t$ induces the trivial Lie bracket on $H^\bullet(\mathcal{I})$. 
\end{proposition}
\begin{proof}
Differentiating identity (\ref{mm'}) once again, we get 
\begin{equation}\label{6.2}
-[m_{(i)}, m_{(j)}]=[m, m_{(ij)}]=\delta m_{(ij)}\,.
\end{equation}
Hence, the bracket $[m_{(i)}, m_{(j)}]$ is a $\delta$-coboundary. By Axiom (hG7) this conclusion extends to arbitrary cup products of $m_{(i)}$'s.  

\end{proof}

Thus, the algebra $H^\bullet(\mathcal{I})$ is generated by cup products of the partial derivatives $m_{(i)}$, so that the general element of $H^\bullet(\mathcal{I})$ is represented  by the cup polynomial
\begin{equation}\label{pol}
\Psi=\sum_{l=0}^L c^{j_1\cdots j_l}m_{(j_1)}\cup m_{(j_2)}\cup\cdots \cup m_{(j_l)}\,,
\end{equation}
where $c^{j_1\cdots j_l}\in k[[t_1,\ldots,t_n]]$. Since all $t$'s are supposed to be even, the graded associative algebra $H^\bullet(\mathcal{I})[-1]$ is commutative in the usual sense and we may assume the coefficients $c^{j_1\cdots j_l}$ to be fully symmetric in permutations of indices. It should be borne in mind that the family of $A_\infty$-structures $m$ enters the polynomial (\ref{pol}) both through the partial derivatives $m_{(j)}$ and through the cup products. 

\begin{theorem}[\cite{Sharapov:2018xxx}]\label{def}
Let $m\in C^\bullet(V)[[t_1,\ldots,t_n]]$ be an $n$-parameter family of $A_\infty$-structures and let $\Psi$ be a cocycle representing an element of $\mathcal{I}^p$. Then one can define an $(n+1)$-parameter family of $A_\infty$-structures $\tilde{m}\in C^\bullet(V)[[t_0,t_1,\ldots,t_n]]$ as a unique formal solution to the differential equation 
\begin{equation}\label{tm}
\tilde{m}_{(0)}=\Psi[\tilde m]
\end{equation}
subject to the initial condition $\tilde{m}|_{t_{0}=0}=m$. Here the new formal parameter $t_{0}$ has degree $1-p$. 
\end{theorem}

\begin{proof}\footnote{
We provide the proof because that in  \cite{Sharapov:2018xxx} contains an unfortunate misprint.} Suppose the cocycle $\Psi$ is given by a polynomial of the form (\ref{pol}). For definiteness we assume that all cup products in $\Psi$ are performed from right to left. 
Clearly, Eq.(\ref{tm}) has a unique formal solution that starts as 
$$
\tilde{m}(t_0)= m+t_0\Psi[m]+O(t^2_0)\,.
$$
Writing identity (\ref{BI}) for $m=\tilde m$,  $A_1=\tilde m_{(j)}$, and $A_2=A$ we get
\begin{equation}\label{auxi}
    [\tilde m, \tilde m_{(j)}\cup A]=(\tilde m\circ \tilde m)_{(j)}\cup A+\tilde m_{(j)}\cup [\tilde m, A]+(\tilde m\circ \tilde m)\{\tilde m_{(j)},A\}\,.
\end{equation}
Here we also used the definition (\ref{CP}). 
On the other hand, differentiating  the cochain $\tilde{m}\circ \tilde{m}$ w.r.t. $t_0$ we obtain
$$
(\tilde m\circ \tilde m)_{(0)}=[\tilde m, \tilde{m}_{(0)}]=[\tilde m, \Psi[\tilde{m}]]\,.
$$
Repeated use of identity (\ref{auxi}) allows us to bring the last equality into the form 
$$
\begin{array}{c}
(\tilde m\circ \tilde m)_{(0)}=\displaystyle 
\sum_{l=0}^L c^{j_1\cdots j_l}\left[ \sum_{k=1}^l \tilde m_{(j_1)}\cup \cdots \cup (\tilde m\circ \tilde m)_{(j_k)}\cup\cdots \cup \tilde m_{(j_l)}\right.\\[8mm]
+(\tilde m\circ\tilde m)\{\tilde m_{(j_1)}, \tilde m_{(j_2)}\cup\cdots\cup\tilde m_{j_l}\}\\[2mm]
\left.+\displaystyle \sum_{k=1}^{l-2}\tilde m_{(j_1)}\cup \cdots\cup \tilde m_{(j_{k})}\cup (\tilde m\circ \tilde m)\{\tilde m_{(j_{k+1})}, \tilde m_{(j_{k+2})}\cup\cdots\cup\tilde m_{(j_{l})}\}\right]\,.
\end{array}
$$
(Cup multiplication is performed from right  to left.) 
We see that $\tilde m\circ \tilde m$ satisfies a linear differential equation. 
With account of the initial condition $\tilde m\circ \tilde m|_{t_0=0}=m\circ m=0$ this means $\tilde m\circ \tilde m=0$. Hence, $\tilde{m}$ defines an $(n+1)$-parameter family of $A_\infty$-structures.  

\end{proof}

We call the deformations of Theorem \ref{def} the {\it inner deformations} of families of $A_\infty$-algebras. Each inner deformation of an $A_\infty$-algebra defines then a deformation of the associated $hG$-algebra $G_t$.

\begin{remark} As a small interlude let us explain in more detail some physical motivations for 
the constructions discussed above and below. Consider a collection of form-fields $\Phi^A$ of various degrees living on a space-time manifold $M$. Regarding these forms as dynamical fields, we can impose on them the equations of motion 
\begin{equation}\label{FDA}
d\Phi^A=\sum_{n\geq 2}l_{A_1\cdots A_2}^A\Phi^{A_1}\wedge\ldots\wedge\Phi^{A_n}\,,
\end{equation}
$d$ being the exterior differential. Formal integrability of the field equations (stemming from $d^2=0$) requires the $l$'s to be structure constants of some minimal $L_\infty$-algebra on the graded vector space dual to the space of form-fields $\Phi^A$. It is not hard to see that every system of PDE can be brought to the form (\ref{FDA}) for a suitable (perhaps infinite) collection of form-fields $\Phi^A$.
Since the $L_\infty$-algebra is minimal, system (\ref{FDA}) admits a consistent truncation setting $l^A_{A_1\cdots A_n}=0$ for all $n>2$. The resulting theory is completely determined by a graded Lie algebra $L$ with structure constants $l^A_{A_1A_2}$. In physical terms, one may view the truncated system as a free field theory, regarding  all the higher terms in r.h.s. of (\ref{FDA}) as interaction. Usually, the Lie algebra $L$, being determined by fundamental symmetries of the theory, is known in advance. Furthermore, in many interesting cases it comes from the commutator of an associative algebra $A$. The main problem is then to reconstruct all  higher structure constants associated with the interaction vertices in (\ref{FDA}). From the algebraic viewpoint, this amounts to deformation of the Lie algebra $L$ to a minimal $L_\infty$-algebra by adding a coherent set of higher multi-brackets. One can also do this at the level of $A_\infty$-algebras: first deform  $A$ to a minimal $A_\infty$-algebra and then apply the symmetrization map (\ref{ln}). 
It is significant that even at the free level the field equations may depend on some free parameters like masses of particles, mixing angles, etc. The machinery of inner deformations  suggests then a systematic way  to convert this dependence into higher structure maps of a minimal $A_\infty$-algebra, that is, into interaction vertices. 

Given a system of PDE, one can wounder about its characteristic cohomology \cite{TSUJISHITA19913, Bryant1995}. Under certain  assumptions \cite{Barnich:2009jy} the characteristic cohomology of  (\ref{FDA}) is generated by on-shell closed differential forms
\begin{equation}\label{J}
J=\sum_{n} J_{A_1\cdots A_n}\Phi^{A_1}\wedge \ldots \wedge \Phi^{A_n}\,,
\end{equation}
 meaning that $dJ=0$ whenever $\Phi^A$ satisfy the equations of motion (\ref{FDA}). In case $\deg J=\dim M-1$, the form $J$ defines a usual conserved current, whose integral over a closed hypersurface in $M$ gives a conserved charge; the forms (\ref{J}) of degree less than $\dim M-1$ define the so-called lower-degree conservation laws, see e.g. \cite{Anderson, Sharapov:2016sgx}.  One may check that the form (\ref{J}) is on-shell closed iff its structure constants $J_{A_1\cdots A_n}$ define and are defined by a cocycle of $L_\infty$-algebra cohomology with trivial coefficients. 
If the $L_\infty$-algebra was obtained by an inner deformation of a graded Lie algebra $L$ (or the underlying associative algebra $A$), one can try to extend this deformation to $L_\infty$-cohomology (resp. $A_\infty$-cohomology). It is this problem that we address in the rest of this paper.  

\end{remark}

\section{Inner deformations of Hochschild (co)cocycles}\label{S6}

In this section, we show that all inner deformations of $A_\infty$-algebras are accompanied by deformations of their Hochschild (co)cycles. 
Let us start with some auxiliary identities. Repeated application of Axioms (hC1-6) gives
$$
    L_{m_{(j_1)}\cup\cdots\cup m_{(j_n)}}\simeq L_{m_{(j_1)}}i_{m_{(j_2)}\cup\cdots\cup m_{(j_n)}}+i_{m_{(j_1)}}L_{m_{(j_2)}\cup\cdots\cup m_{(j_n)}}$$
    $$
    \simeq \displaystyle \sum_{k=1}^ni_{m_{(j_1)}}\cdots i_{m_{(j_{k-1})}}L_{m_{(j_k)}}i_{m_{(j_{k+1})}}\cdots i_{m_{(j_n)}}$$
    $$
    \simeq\displaystyle \sum_{k=1}^n\Big (i_{m_{(j_1)}}\cdots i_{m_{(j_{k-1})}}i_{m_{(j_{k+1})}}L_{m_{(j_k)}}i_{m_{(j_{k+2})}}\cdots i_{m_{(j_n)}}$$
    $$
    -i_{m_{(j_1)}}\cdots i_{m_{(j_{k-1})}}
    i_{[m_{(j_{k})},m_{(j_{k+1})}]}i_{m_{(j_{k+2})}}\cdots i_{m_{(j_n)}}\Big)$$
    $$
    \simeq\displaystyle \sum_{k=1}^n \Big (i_{m_{(j_1)}}\cdots i_{m_{(j_{k-1})}}i_{m_{(j_{k+1})}}\cdots i_{m_{(j_n)}}L_{m_{(j_k)}}$$
    $$
    -\displaystyle \sum_{l=1}^{n-k}[L_m, i_{m_{(j_1)}}\cdots i_{m_{(j_{k-1})}}i_{m_{(j_{k+1})}}\cdots i_{m_{(j_{k+l-1})}}i_{m_{(j_k,j_{k+l})}}i_{m_{(j_{k+l+1})}}\cdots i_{m_{(j_n)}}]\Big).
$$
Here we also used Eq. (\ref{6.2}). Replacing the monomial $m_{(j_1)}\cup\cdots\cup m_{(j_n)}$ above with a general cup polynomial (\ref{pol}) and restoring the strict equality sign,  we can write 
\begin{equation}\label{LD}
    L_{\Psi}=[L_{ m}, W] + \sum_{j=1}^nV^{j} L_{{m}_{(j)}}
\end{equation}
for some $W=W[m]\in C^\bullet(V)[[t_1,\ldots, t_n]]$ and\footnote{Recall that $c^{j_1\ldots j_l}$ are symmetric in upper indices.}
$$V^{j}=\sum_{l=1}^{L}lc^{j_1\ldots j_{l-1}j}i_{m_{(j_1)}}\cdots i_{m_{(j_{l-1})}}\,.$$ 
It is important that $[L_m,V^{j}]=0$.  

Suppose now that deformation generated by the flow (\ref{tm}) extends to deformation of 
the corresponding  Hochschild cycles, that is, for every cycle $\alpha\in C_\bullet(V, m)$ there exists a cycle $\tilde \alpha\in C_\bullet(V,\tilde m)$ such that $\tilde \alpha|_{t_0=0}=\alpha$. Differentiating the boundary $\tilde \partial \tilde \alpha=L_{\tilde m}\tilde \alpha$ w.r.t. $t_0$ we get
$$
    (\tilde \partial \tilde \alpha)_{(0)}=(L_{\tilde m}\tilde \alpha)_{(0)}=L_{\tilde m_{(0)}}\tilde \alpha+L_{\tilde m}\tilde \alpha_{(0)}=L_\Psi \tilde \alpha+L_{\tilde m}\tilde \alpha_{(0)}\,, 
$$
where $\Psi=\Psi[\tilde m]$. With account of (\ref{LD}) this is equivalent to 
\begin{equation}\label{La}
(\tilde \partial \tilde \alpha)_{(0)}=L_{\tilde m} W\tilde \alpha + WL_{\tilde m} \tilde \alpha + \sum_{j=1}^nV^j L_{m_{(j)}}\tilde \alpha+L_{\tilde m}\tilde \alpha_{(0)}
\end{equation}
for some $W=W[\tilde m]$.
The r.h.s of this equality suggests to impose the following differential equation on $\tilde \alpha$:
\begin{equation}\label{a0}
    \tilde \alpha_{(0)}=\sum_{j=1}^nV^j\tilde \alpha_{(j)} - W\tilde \alpha\,.
\end{equation}
Then Eq. (\ref{La}) takes the form
\begin{equation}\label{LW}
({\tilde\partial  }\tilde \alpha)_{(0)}=W({\tilde \partial }\tilde \alpha)+\sum_{j=1}^nV^j({\tilde\partial  }\tilde \alpha)_{(j)}\,.
\end{equation}
We see that once a chain $\tilde \alpha$ satisfies Eq. (\ref{a0}), its boundary ${\tilde \partial}\tilde \alpha$ obeys the linear homogeneous equation  (\ref{LW}). For the zero initial condition ${\tilde \partial }\tilde \alpha|_{t_0=0}=\partial \alpha=0$, this means that $\tilde \alpha$, being uniquely defined by Eq. (\ref{a0}),  is a cycle of $C_\bullet(V,{\tilde m})$. We thus arrive at the main result of this paper. 

\begin{theorem}\label{T61}
Let $\tilde m$ be an inner deformation of a family of $A_\infty$-structures $m$. Then each Hochschild cycle $\alpha$ of $C_\bullet(V,m)$ deforms to a cycle $\tilde \alpha$  of $C_\bullet(V,\tilde m)$. 
\end{theorem}
This means stability of Hochcshild (co)homology under inner deformations of $A_\infty$-algebras. 
Let us write down a more explicit formula for the flow (\ref{a0}) in a special case of quadratic  deformation flow (\ref{tm}), namely, 
\begin{equation}\label{m012}
    \tilde m_{(0)}=\tilde m_{(1)}\cup \tilde m_{(2)}\,,
\end{equation}
$\tilde m(0, t_1,t_2)=m(t_1,t_2)$ being a two-parameter family of $A_\infty$-structures. 
The associated flow (\ref{a0}) on chains takes then the form 
$$
    \tilde \alpha_{(0)}=i_{\tilde m_{(1)}}\tilde\alpha_{(2)}+i_{\tilde m_{(2)}}\tilde \alpha_{(1)}+\big (i_{\tilde m_{(1,2)}}+T(\tilde m_{(1)},\tilde m_{(2)})-L(\tilde m_{(1)},\tilde m_{(2)})\big)\tilde \alpha,
$$
where the endomorphisms $L$ and $T$ are defined by Eqs. (\ref{LAA}) and (\ref{TAB}). 

Since $\mathcal{I}$ consists of $\delta$-cocycles, each $\Psi\in \mathcal{I}$ generates an endomorphism of the $k[[t_1,\ldots,t_n]]$-vector space $HH_\bullet(m,V)$. At the level of chains,  the endomorphism is defined by the assignment $\alpha\mapsto i_{\Psi} \alpha$ for all $\alpha\in C_\bullet(V,m)$.   This makes the space of Hochschild homology into a module over the commutative algebra $H^\bullet(\mathcal{I})[-1]$. Generally, the algebra $H^\bullet(\mathcal{I})[-1]$ is not free, and the $\delta$-cohomology classes of some polynomials (\ref{pol}) may vanish identically, thereby defining relations in $H^\bullet(\mathcal{I})[-1]$.

\begin{theorem}\label{T62}
Every relation $[\Psi]= 0$ in  $H^\bullet(\mathcal{I})[-1]$ gives rise to an endomorphism $h_\Psi$ of the space $HH_\bullet(V,m)$ . 
\end{theorem}

\begin{proof}
If the polynomial (\ref{pol}) defines a relation, then $\Psi=\delta U$ for some $U\in C^\bullet(V)[[t_1,\ldots,t_n]]$. Let $\alpha$ be a cycle 
of  $C_\bullet(V,m)$.
By (\ref{Lrep}) and (\ref{LD})
$$
\begin{array}{rcl}
    0&=&L_{\Psi-\delta U}\alpha =L_\Psi\alpha-[L_m, L_U]\alpha\\[2mm]
    &=&\displaystyle [L_m, W-L_U]\alpha +\sum_{j=1}^nV^jL_{m_{(j)}}\alpha\\[2mm]
    &=&\displaystyle L_m\Big( W\alpha -L_U\alpha -\sum_{j=1}^nV^j\alpha_{(j)}\Big)\,.
\end{array}
$$
Hence, the expression in the round brackets is a cycle and we can define the desired endomorphism $h_\Psi$ by the formula
$$
h_\Psi[\alpha]=\Big[ W\alpha-L_U\alpha-\sum_{j=1}^nV^j\alpha_{(j)}\Big]\,.
$$
\end{proof}
Thus, the algebra $H(\mathcal{I})[-1]$ gives rise to a set of operations on $HH_\bullet(V,m)$, both through the action of its generators and through  relations among them.

Theorems \ref{T61} and \ref{T62} have obvious counterparts for the cochain complex $C^\bullet (V,m)$. 
In particular, if $S$ is a cocycle of  $C^\bullet(V,m)$  and the inner deformation is generated by the quadratic flow (\ref{m012}), then the flow (\ref{a0}) is replaced with the following one:
$$
    \tilde S_{(0)}=\tilde S_{(1)}i_{m_{(2)}}+S_{(2)}i_{m_{(1)}}+\tilde S\big(i_{m_{(1,2)}} +L(m_{(1)},m_{(2)})+T(m_{(2)},m_{(1)})\big)\,.
$$
It is easy to see that $\tilde S{\tilde \partial}=0$ whenever $S\partial=0$. 

\section{An application to dg-algebras}
By way of illustration, let us apply the above machinery of inner deformations to the case of dg-algebras. Recall that a dg-algebra is given by a triple 
$(A,\cdot, d)$, where $A=\bigoplus A_n$ is a $\mathbb{Z}$-graded vector space over $k$ endowed with an associative dot product and a differential $d:A_n\rightarrow A_{n-1}$. We also assume the dg-algebra $A$ to enjoy a closed trace, i.e., a  linear map $\mathrm{Tr}: A\rightarrow k$ obeying the two conditions
\begin{equation}\label{tr}
\mathrm{Tr}(ab)=(-1)^{|a||b|}\mathrm{Tr}(b a)\,,\qquad \mathrm{Tr}(da)=0\,.
\end{equation}
For our purposes, it is convenient to regard the graded space $V=A[1]$ as an $A_\infty$-algebra with $m=m_2$, where $m_2(a,b)=(-1)^{|a|-1}ab$. (There is no way to include $d$ into the $A_\infty$-structure $m$, as it has degree $-1$; notice, however, that $[m,d]=0$.) 
Extending the trace functional from $A$ to the whole tensor algebra $TA$ by zero, we can write  Rels. (\ref{tr}) as $\mathrm{Tr}L_{m}=0$ and $\mathrm{Tr}L_d=0$. Define the sequence of multi-linear maps on $V$ by setting
$$
S_0=\mathrm{Tr}\,,\qquad S_{k+1}=S_k i_d\,,\qquad  k=0,1,2,\ldots,
$$
or, more explicitly,
\begin{equation}\label{Str}
    S_k(a_0,a_1,\ldots,a_k)=\mathrm{Tr}(a_0da_1\cdots da_k)\,.
\end{equation}
We claim that all $S_k$'s are cocycles of the Hochschild complex $C^\bullet(V, m)$. Indeed, 
$$
S_k\partial =\mathrm{Tr}(i_d)^k\partial =\mathrm{Tr}L_m (i_d)^k=0\,.
$$
Here we used (hC1).
A more familiar  form of the cocycle condition above is
\begin{equation}\label{dS}
\begin{array}{l}
    \displaystyle \sum_{j=0}^{k}(-1)^{\bar a_0+\cdots+\bar a_j} S_k(a_0,\ldots,a_ja_{j+1},\ldots, a_{k+1})\\[5mm]
+(-1)^{\bar a_{k+1}(\bar a_0+\cdots +\bar a_{k}+1)} S_k(a_{k+1}a_0, a_1,\ldots, a_{k})=0\,,
\end{array}
\end{equation}
where $\bar a=|a|-1$ is the degree of $a$ as an element of $V=A[1]$. Using the explicit expression (\ref{Str}),  one can easily verify the cyclicity property 
\begin{equation}\label{cycS}
S_k(a_0,a_1,\ldots,a_k)=(-1)^{\bar a_0(\bar a_1+\cdots+\bar a_k)}S_k(a_1,\ldots,a_k,a_0)\,.
\end{equation}
 Together Rels. (\ref{dS}) and (\ref{cycS}) identify $S_k$ as a cyclic cocycle representing a class of $HC^k_\lambda(A)$. 

Suppose now that the dg-algebra $A$ admits a formal deformation $A_t=A[[t]]$, the parameter $t$ being of degree zero.  It is known that 
neither cyclic nor Hochschild (co)homology is homotopy invariant.  Nevertheless, we assume that the trace (\ref{tr}), representing a class of $HC_\lambda^0(A)=HH^0(A)$,  survives deformation, so that the cyclic cocycles (\ref{Str}) make sense for the deformed dg-algebra $A_t$ as well.

The results of the previous section suggest that a kind of homotopy invariance does take place for the Hochschild (co)homology, if we restrict ourselves to inner deformations of families. Although there is no way to define  inner deformations of $A_t$ as an associative algebra, such deformations are possible in a larger category of $A_\infty$-algebras \cite{Sharapov:2018xxx}.  The construction goes as follows. First, we introduce an auxiliary formal variable $v$ of degree two and define the dg-algebra $A_t[[v]]$, just extending the original differential and product by $k[[v]]$-linearity.  
Multiplying the differential $d$ by $v$ gives then a differential of degree one on $A_t[[v]]$. 
This allows us to regard $A_t[[v]]$ as a two-parameter family of $A_\infty$-algebras with $m=m_1+m_2$, where 
$$m_1(a)=vd a\,,\qquad   m_2(a,b)=(-1)^{\bar a}ab\,.$$ 
The corresponding algebra of inner deformations $\mathcal{I}$ is generated by the cocycles $m_{(t)}$ and $m_{(v)}=d$.  Now we define the sequence of cocycles 
\begin{equation}\label{psi}
   \Psi_0=m_{(t)}\,,\qquad  \Psi_{n+1}= d\cup \Psi_{n}\,, \qquad n=0,1,2,\ldots \,,
\end{equation}
all are of degree one. According to Theorem \ref{def}, each cocycle $\Psi_n$ gives rise to the deformation flow 
$\tilde m_{(s)}=\Psi_n[\tilde m]$ on the space of $A_\infty$-structures,
with $s$ being a new formal variable of degree zero.  Since the parameter $v$ plays an auxiliary role in our construction, we can 
finally put it to zero to get a family of $A_\infty$-structures $\bar m=\tilde{m}|_{v=0}$ on $V$ parameterized by $t$ and $s$; both the parameters are of degree zero.  The family starts as 
\begin{equation}\label{bm}
 \bar m=m_2 +s \bar \Psi_n+ O(s^2)\,,   
\end{equation}
where
$$
   \bar \Psi_n(a_0,a_1,\ldots,a_{n+1})=(-1)^{n+\bar a_n}da_0\cdots da_{n-1}(a_n\diamond a_{n+1})
$$
and $\diamond$ stands for the $t$-derivative of the dot product in $A_t$. By construction, $\bar \Psi_n$ is a Hochschild cocycle  representing a class of $HH^{n+2}(A_t,A_t)$.  If the class happens to be nontrivial, then (\ref{bm}) defines a nontrivial deformation of the associative  algebra $A_t$ in the category of $A_\infty$-algebras\footnote{For $\Psi_0$, the deformation boils down to the trivial shift  $t\mapsto t+s$.}. Notice that the resulting $A_\infty$-structure $\bar m$ is { minimal}, for $\bar m_1=0$. We denote the corresponding differential by $\bar \partial=L_{\bar m}$.

As shown in Sec. \ref{S6}, every inner deformation of an $A_\infty$-algebra extends to  Hochschild (co)cycles. 
\begin{proposition}
The inner deformation (\ref{bm}) generates the following first-order deformation of the cocycles (\ref{Str}):
\begin{equation}\label{Sk}
\bar S_k= S_k+ s(S_{k+n})_{(t)}+O(s^2)\,.
\end{equation}
\end{proposition}
\begin{proof}
Computation  is somewhat simplified due to specific properties of the  cocycles (\ref{Str}). First of all, we have 
$$
S_kL_{\bar\Psi_{n+1}}=S_kL_{d\cup \bar \Psi_n}=S_k(i_dL_{\bar \Psi_n}+L_{d}i_{\bar \Psi_n} +[\partial, L(d, \bar \Psi_n)])\,.
$$
Actually, the last two terms on the right vanish. It is immediate that
$$
S_kL_d=\mathrm{Tr}L_d(i_d)^k=0\,,\qquad S_k\partial=0\,. 
$$
The remaining term $S_kL(d,\bar\Psi_n)\partial$ of the commutator requires more work. Since $d$ is a unary map, it follows from  the definition (\ref{LAA}) that 
\begin{equation}\label{Laa}
\begin{array}{c}
    L(d,\bar\Psi_n)(a_0,\ldots, a_p)\\[3mm]
   \displaystyle  =\sum_{0\leq i<j\leq p-n-1}(-1)^\kappa(a_0,\ldots, da_i,\ldots, \bar\Psi_n(a_{j},\ldots, a_{j+n+1}),\ldots, a_p),
 \end{array}
\end{equation}
whence
$$
 S_kL(d,\bar\Psi_n)(a_0,\ldots, a_{k+n+1})
 $$
 $$
 =\mathrm{Tr}\Big( d\sum_{j=1}^k(-1)^\kappa a_0da_1\cdots  d\bar\Psi_n(a_{j},\ldots, da_{j+n+1})\cdots da_{k+n+1}\Big)=0\,.
$$
Therefore,
$$
    S_kL_{\bar \Psi_{n+1}}=S_{k+1}L_{\bar \Psi_{n}}=\ldots =S_{k+n+1}L_{\bar \Psi_0}\,.
$$
Now, verifying the cocycle condition for (\ref{Sk}), we find 
$$
\begin{array}{rcl}
    \bar S_k\bar \partial&=&\big(S_k+s(S_{k+n})_{(t)}+O(s^2)\big)\big(L_{m_2}+sL_{\bar \Psi_n}+O(s^2)\big)\\[3mm]
    &=&s\big((S_{k+n})_{(t)}L_{m_2}+S_kL_{\bar\Psi_n}\big)+O(s^2)\\[3mm]
     &=&s\big((S_{k+n})_{(t)}L_{m_2}+S_{k+n}L_{\bar\Psi_0}\big)+O(s^2)\\[3mm]
     &=&s(S_{k+n}L_{m_2})_{(t)} +O(s^2)=O(s^2)\,.
    \end{array}
$$
The proof is complete. 
\end{proof}
 Notice that the first-order deformation (\ref{Sk}) enjoys cyclicity (\ref{cycS}) as well. It is not clear whether this property holds true in higher orders in $s$. Should this be the case the deformed cocycles $\bar S_k$ would define elements of the cyclic cohomology group $HC^\bullet_\lambda(V, \bar m)$.  

As mentioned in the Introduction, inner deformations of $A_\infty$ algebras play a crucial role in construction of a new class of classical integrable theories \cite{Sharapov:2019vyd} and in applications to Chern--Simons matter theories and $3d$ bosonization duality \cite{Sharapov:2018kjz}. In this context, the most interesting are deformations of dg-algebras that are generated by the first Hochschild cocycle $\Psi_1$  of  (\ref{psi}).
Furthermore, it is demanded  that the deformed cocycles (\ref{Sk}) be cohomologous to cyclic ones to be interpreted in terms of correlators of higher-spin currents.  In the rest of this paper, we show that the inner deformation generated by $\Psi_1$ does respect cyclicity under the two assumptions: 
\begin{enumerate}
    \item[(I)] $A_t=\bigoplus_{n\in \mathbb{Z}} (A_t)_n$ is a family of unital dg-algebras with $(A_t)_n=0$ for all $n> 1$,
    \item[(II)]the differential $d: (A_t)_n\rightarrow (A_t)_{n-1}$ is independent of $t$, i.e., $d_{(t)}=0$.
\end{enumerate}
Again, this is enough for the applications mentioned above. Even with these technical restrictions we need some preparation.  

The main difficulty here is that the operator $1-\lambda$ features no good algebraic properties in $hG$-calculus and the cyclicity  of deformed cochains is hard to control. Fortunately, there is an equivalent definition of cyclic cohomology, which fits perfectly well with  $hG$-calculus. Following \cite{Getzler1990A}, \cite{Getzler93}, we introduce a new formal variable $u$ of degree $2$ and endow the space $\bar C_\bullet(V)[u]$ of polynomials in $u$ with the differential $D=\partial-uB$, where $B$ is the Connes--Rinehart operator (\ref{B}).
Regarding $\bar C_\bullet(V)[u]$ as a $k[u]$-module, we define the dual cochain complex $\bar C_\bullet(V)^\ast [u]$.
Assumption (I) above implies that the space $V$ is non-positively graded. As a consequence, the complex $\bar C_\bullet(V)^\ast[u]$ is concentrated in non-negative degrees and each $p$-cochain $c\in \bar C_\bullet(V)^\ast[u]$ is given  by a polynomial
$$
c=c_p+c_{p-2}u+\cdots +c_{p-2q}u^q\,,\qquad |c_i|=i,
$$
of degree $q\leq p/2$. We denote the cohomology of this complex by $HC_{^{[u]}}^\bullet(V,m)$.  The following isomorphism is well known:
\begin{equation}\label{iso}
    HC_{^{[u]}}^\bullet (V,m)\simeq HC^\bullet_\lambda(V,m)\,.
\end{equation}
See e.g. \cite{hood1986}, \cite{Volkov}. The isomorphism allows us to work entirely in terms of $hG$-calculus' operations. It is natural to think of $D$ as a perturbation of the Hochschild differential $\partial$. We also  introduce the operator 
$I_A=i_A-uS_A$, a perturbation of the contraction $i_A$.  Then, the following relations hold: 
$$
\begin{array}{cl}
    (1) &D^2=0\,,\\[1mm]
    (2)& [D, L_A]=L_{\delta A}\,,\\[1mm]
    (2)&uL_A=[D,I_A]+I_{\delta A}\,,\\[1mm]
    (3)&L_{A\cup B}=L_AI_B-(-1)^{|A|}I_AL_B\\[1mm]
      &\hspace{11mm}-(-1)^{|A|}[D, L(A,B)]+(-1)^{|A|}L(\delta A, B)+L(A,\delta B)\,,\\[1mm]
    (4)&(-1)^{|A|}[L_A,I_B]=I_{[A,B]}+[D,T(A,B)]-T(\delta A,B)-(-1)^{|A|}T(A,\delta B),\\[1mm]
    (5)&uL(A,B)= (-1)^{|A|-1}I_{A\cup B}+(-1)^{|A|}I_AI_B+[D,I(A,B)]\\[1mm]
    & \hspace{65mm}+I(\delta A,B)+(-1)^{|A|}I(A,\delta B)\,.
\end{array}
$$
The reader can verify these identities either  directly, using the definitions of Sec. \ref{S3}, or consult the papers \cite{Getzler93}, \cite{DomKow}.  
Comparing Rels. (1-5)  with Axioms (hC1-6), we see that the operations $D$, $L$, and $I$ do not define a precalculus because of the `curvature' terms $uL_A$ and $uL(A,B)$ in (2) and (5). Nonetheless, these relations appear to be very helpful  in practical calculations, as we will see shortly. 

Given a family of $A_\infty$-structures $m=vd+m_2$ as above, we consider the inner deformation  generated by the cocycle 
$$
\Psi_1=m_{(v)}\cup m_{(t)}=d\cup m'_2.
$$
Hereinafter the prime stands for the derivative  with respect to $t$.  
Notice that the r.h.s. of the  equation   $\tilde m_{(s)}=\Psi_1[\tilde m]$ does not depend on $v$. As a result, its formal solution has the form 
$$
    \tilde m=v d+\bar m= m +s(d\cup m'_2)+O(s^2)\,,
$$
where the cochain  $\bar m$ is independent of $v$ either. It follows  from $[\tilde{m},\tilde{m}]=0$ that $\bar m=\bar{m}(t,s)$ is a two-parameter family of minimal $A_\infty$-structures which commutes with the differential $d$, i.e.,
$$
[\bar m,\bar m]=0\,,\qquad [d, \bar m]=0\,.
$$
We are interested in  cocycles $\hat S\in \bar C_\bullet(V)^\ast[u]$ of the differential $\bar D=\bar \partial-uB$. 
\begin{proposition}\label{P72}
With notation and assumptions above the deformation flow 
\begin{equation}\label{mS}
    \bar m_{(s)}= d\cup\bar m'\,,\qquad \hat S_{(s)}= \hat S'I_d-\hat S L(d,\bar m')
\end{equation}
preserves the conditions 
\begin{equation}\label{SDL}
   \hat S\bar D=0\,,\qquad \hat SL_d=0\,.
\end{equation}
In other words, the flow maps the space of $L_d$-closed $\bar D$-cocycles into itself. 
\end{proposition}
\begin{proof} Differentiating the cochain $\hat S\bar D$, we  get
$$
\begin{array}{rcl}
(\hat S\bar D)_{(s)}&=&\hat S_{(s)}\bar D+\hat S \bar D_{(s)}=\hat S_{(s)}\bar D+\hat S L_{\bar m_{(s)}}=\hat S_{(s)}\bar D+\hat S L_{d\cup \bar m'}\\[3mm]
&=&\hat S_{(s)}\bar D+\hat S \big (L_d I_{\bar m'}+I_dL_{\bar m'}+[\bar D, L(d,\bar m')]\big )\,.
\end{array}
$$
Since $d$ is a unary operator, $T(d, \bar m' )=0$. Furthermore, $[d, \bar m']=[d,\bar m]'=0$, by assumption (II).  Combining these observations with the flow (\ref{mS}) for $\hat S$, we can proceed as follows 
\begin{equation}\label{SDs}
\begin{array}{rcl}
(\hat S\bar D)_{(s)}&=&\hat S_{(s)}\bar D+\hat S \big (L_d I_{\bar m'}+L_{\bar m'}I_d+[\bar D, L(d,\bar m')]\big )\\[3mm]
&=&\hat S' I_d\bar D+\hat SL_dI_{\bar m'}+\hat SL_{\bar m'}I_d+\hat S\bar D L(d,\bar m')\\[3mm]
&=&\hat S' \bar DI_d +u\hat S' L_d +\hat SL_dI_{\bar m'}+\hat SL_{\bar m'}I_d+\hat S\bar D L(d,\bar m')\\[3mm]
&=&(\hat S \bar D)'I_d +u(\hat S L_d)' +(\hat SL_d)I_{\bar m'}+(\hat S\bar D) L(d,\bar m')\,.
\end{array}
\end{equation}
In the same way we find 
$$
(\hat S L_d)_{(s)}=\hat S_{(s)}L_d=\hat S'I_dL_d-\hat SL(d,\bar m')L_d =\hat S'L_dI_d-\hat SL(d,\bar m')L_d\,.
$$
The operator $L(d,\bar m')$ acts according to  Eq. (\ref{Laa}) with $\bar \Psi_n$ replaced by $\bar m'$.  Since the differential 
$d$ commutes with $\bar m'$, one can easily see that $[L(d, \bar m'), L_d]=0$. Therefore, we can write 
\begin{equation}\label{Sds}
    (\hat S L_d)_{(s)}=(\hat SL_d)'I_d-(\hat SL_d)L(d,\bar m')\,.
\end{equation}
Taken together Eqs. (\ref{SDs}) and (\ref{Sds}) imply the invariance of the algebraic constraints (\ref{SDL}) under the flow (\ref{mS}).
\end{proof}

Returning  to the Hochschild cocycles (\ref{Str}) of the dg-algebra $A$, one can see that all of them are annihilated by the Connes--Rinehart operator, $S_kB=0$, and hence, by the differential $D=\partial-uB$. Furthermore, they are annihilated by the Lie derivative $L_d$ as well. By Proposition \ref{P72},  the deformed cocycles $\hat S_k$ define some cohomology classes of $HC_{^{[u]}}^\bullet(V,\bar m)$. It is apparent that the leading term $\bar S_k=\hat S_k|_{u=0}$   coincides with the Hochschild cocycle (\ref{Sk}) and should be cohomologous to a cyclic cocycle  of $C^\bullet_\lambda(V,\bar m)$ due to the isomorphism (\ref{iso}). In other words, the cohomology classes of $\bar S_k$ belong to the preimage of the canonical homomorphism $I: HC_\lambda^\bullet(V, \bar m)\rightarrow HH^\bullet(V, \bar m)$. 
Let us denote a cyclic representative of the class $[\bar S_k]\in HC^\bullet_\lambda(V,\bar m)$ by $\bar S^\lambda_k$. By definition, $\bar S^\lambda_k(1-\lambda)=0$.

Now let $\bar l$ be the $L_\infty$-structure on $V$ associated with the $A_\infty$-structure $\bar m$  according to Theorem \ref{Th10}. 
It is easy to see that the  symmetrization map 
\begin{equation}\label{symm}
(\vartheta S) (a_0\vee\ldots\vee a_n)=\sum_{\sigma\in S_n}(-1)^\kappa S(a_0,  a_{\sigma(1)}, \ldots, a_{\sigma(n)})
\end{equation}
commutes with the differentials $D_{\bar l}$ in $SV$ and $\bar\partial$ in $C_\bullet^\lambda(V, \bar m)$, inducing thus a homomorphism in cohomology:
$$
\vartheta^\ast:  HH_\lambda^\bullet(V,\bar m)\rightarrow H^\bullet(V,\bar l)\,.
$$
Notice that the map $\vartheta$ is well-defined as the cochain $S$ enjoys cyclicity.  

Applying the symmetrization map (\ref{symm}) to the cyclic cocycles $\bar S^\lambda_k$ yields the sequence of $D_{\bar l}$-cocycles  $ \Phi_k=\vartheta (\bar S^\lambda_k)$, which represent some classes of $L_\infty$-cohomology $H^\bullet (V,\bar l)$, see Sec. \ref{S4}. 
By definition, each cocycle $\Phi_k$ is given by a collection of multi-linear maps  $\Phi_k^n: S^{n+1}V\rightarrow k$ with $n\geq k$. The functions $\Phi_k^n(a_0,\ldots, a_n)$, being  symmetric, are completely specified by their values  $\Phi_k^n(a,\ldots,a)$ on coinciding arguments $a\in V$. This allows us to encode the cocycle $\Phi_k=\{\Phi_k^n\}_{n=k}^\infty$ by a single function on $V$:
\begin{equation}\label{FF}
\Phi_k(a)=\sum_{n=k}^\infty \Phi_k^n(\underbrace{a,\ldots, a}_{n+1})\,.
\end{equation}
Similarly, one can define  a generating function for the $L_\infty$-structure $\bar l$. One just regards $V$ as a graded manifold and endow it with the homological vector field $Q$ whose action on `coordinate functions' $a\in V$ is given by 
$$
Qa=\sum_{n=2}^\infty \bar l_n(\underbrace{a,\ldots,a}_{n})\,.
$$
In this geometric language the cocycle condition $\Phi_k D_{\bar l}=0$ amounts to the $Q$-invariance of the function (\ref{FF}), i.e., $Q\Phi_k=0$. 

We conclude this paper by presenting an example where the deformation flow $\Psi_1$ for the cocycles $\Phi_k$ can be integrated explicitly.  Let $A=A_0\oplus A_1$ be a one-parameter family of dg-algebras with trace. Suppose that the differential $d: A_1\rightarrow A_0$ defines an isomorphism of vector spaces and the trace $\mathrm{Tr}: A\rightarrow k$ vanishes on $A_0$.  As is shown in \cite{Gerasimenko2022SlightlyBH}, the cocycles (\ref{FF}) are given then by the formula  
\begin{equation}\label{Fa}
\Phi_k(a) =\sum_{n=0}^\infty \frac {1}{k+n+1} \mathrm{Tr}(a\underbrace{ da \cdots da}_{k+n})\Big\rfloor_{s^n} \quad\qquad \forall a\in A_1\,,
\end{equation}
cf. Eq.(\ref{Sk}). In this expression, the parameter $t$, which enters both the trace and the product,  is replaced with $t+s$ and the symbol  $\big\rfloor_{s^n}$ instructs one to take only the terms that are proportional to $s^n$.  We can also assemble all $\Phi_k(a)$ into a single generating function $\Phi(a)=\sum_{k\geq 0} \nu^{-k-1}\Phi_k(a)$ with an auxiliary parameter $\nu$ to find 
\begin{equation}\label{FaB}
\Phi(a) =\mathrm{P.p.} \,\mathrm{Tr}\big (d^{-1}\ln (1-\nu^{-1} da)\big)\Big|_{t\rightarrow t+s \nu}\,.
\end{equation}
Here $\mathrm {P.p.} $ stands for the principal part of the Laurent series in $\nu$. 

The existence of compact expression (\ref{FaB}) looks quite surprising because no simple formula is known for the deformed $A_\infty$-structure $\bar m$. Note also a striking similarity of (\ref{FaB}) with  the partition function of a one-loop exact QFT. In the applications to Chern--Simons matter vector models and three-dimensional bosonization duality, one identifies $V$ with the space of conformally invariant tensor currents \cite{Gerasimenko2022SlightlyBH}.  
Conceivably all these might point to a deeper relationship of the cocycles (\ref{Fa}) with one-loop exact QFT models yet to be found.

\section*{Acknowledgements}  
We are grateful to Isaac Newton for his profound remark ``Data aequatione quotcunque fluentes quantitae involvente fluxiones invenire et vice versa''.
The work of A. Sh. was supported by the Ministry of Science and Higher Education of the Russian Federation (project No. FSWM-2020-0033) and by  the Foundation for the Advancement of Theoretical Physics and Mathematics ``BASIS''.
The work of E.S. has received funding from the European Research Council (ERC) under the European Union’s Horizon 2020 research and innovation programme (grant agreement No. 101002551). 

\providecommand{\href}[2]{#2}\begingroup\raggedright\endgroup


\begin{thebibliography}{10}

\bibitem{stasheff1998secret}
J.~Stasheff, ``{The (secret?) homological algebra of the Batalin-Vilkovisky
  approach},'' {\em Contemporary mathematics} {\bfseries 219} (1998) 195--210.

\bibitem{jurvco2019algebras}
B.~Jur{\v{c}}o, L.~Raspollini, C.~S{\"a}mann, and M.~Wolf,
  ``{$L_\infty$-Algebras of Classical Field Theories and the
  Batalin--Vilkovisky Formalism},'' {\em Fortschritte der Physik} {\bfseries
  67} no.~7, (2019) 1900025.

\bibitem{Zwiebach:1992ie}
B.~Zwiebach, ``{Closed string field theory: Quantum action and the B-V master
  equation},'' {\em Nucl. Phys.} {\bfseries B390} (1993) 33--152,
\href{http://arxiv.org/abs/hep-th/9206084}{{\ttfamily arXiv:hep-th/9206084
  [hep-th]}}.

\bibitem{Kajiura:2003ax}
H.~Kajiura, ``{Noncommutative homotopy algebras associated with open
  strings},'' {\em Rev. Math. Phys.} {\bfseries 19} (2007) 1--99,
\href{http://arxiv.org/abs/math/0306332}{{\ttfamily arXiv:math/0306332
  [math-qa]}}.

\bibitem{kajiura2006homotopy}
H.~Kajiura and J.~Stasheff, ``Homotopy algebras inspired by classical
  open-closed string field theory,'' {\em Commun. in math. phys.} {\bfseries
  263} no.~3, (2006) 553--581.

\bibitem{Kontsevich:1997vb}
M.~Kontsevich, ``{Deformation quantization of Poisson manifolds. 1.},'' {\em
  Lett. Math. Phys.} {\bfseries 66} (2003) 157--216,
\href{http://arxiv.org/abs/q-alg/9709040}{{\ttfamily arXiv:q-alg/9709040
  [q-alg]}}.

\bibitem{Sharapov:2019vyd}
A.~Sharapov and E.~Skvortsov, ``{Formal Higher Spin Gravities},'' {\em Nucl.
  Phys.} {\bfseries B941} (2019) 838--860,
\href{http://arxiv.org/abs/1901.01426}{{\ttfamily arXiv:1901.01426 [hep-th]}}.

\bibitem{Sharapov:2018kjz}
A.~Sharapov and E.~Skvortsov, ``{$A_\infty$ algebras from slightly broken
  higher spin symmetries},''
  \href{http://dx.doi.org/10.1007/JHEP09(2019)024}{{\em JHEP} {\bfseries 09}
  (2019) 024},
\href{http://arxiv.org/abs/1809.10027}{{\ttfamily arXiv:1809.10027 [hep-th]}}.

\bibitem{Maldacena:2012sf}
J.~Maldacena and A.~Zhiboedov, ``Constraining conformal field theories with a
  slightly broken higher spin symmetry,'' {\em Classical and Quantum Gravity}
  {\bfseries 30} no.~10, (2013) 104003.

\bibitem{Gerasimenko2022SlightlyBH}
P.~Gerasimenko, A.~A. Sharapov, and E.~D. Skvortsov, ``Slightly broken higher
  spin symmetry: general structure of correlators,'' {\em JHEP} {\bfseries
  2022} (2022) 1--31.

\bibitem{Sharapov:2018xxx}
A.~Sharapov and E.~Skvortsov, ``{Cup product on $A_\infty$-cohomology and
  deformations},'' {\em J. Noncommut. Geom.} {\bfseries 15} no.~1, (2021)
  223--240.

\bibitem{gel1989variant}
I.~M. Gel'fand, Y.~L. Daletskii, and B.~L. Tsygan, ``On a variant of
  noncommutative differential geometry,'' {\em {Doklady Akademii Nauk}}
  {\bfseries 308} no.~6, (1989) 1293--1297.

\bibitem{Getzler93}
E.~Getzler, ``{Cartan homotopy formulas and the Gauss--Manin connection in
  cyclic homology},'' in {\em Quantum deformations of algebras and their
  representations}, vol.~7, pp.~65--78.
\newblock Bar-Ilan University, Ramat-Gan, 1993.

\bibitem{nest1999cohomology}
R.~Nest and B.~Tsygan, ``On the cohomology ring of an algebra,'' in {\em
  Advances in geometry}, pp.~337--370.
\newblock Springer, 1999.

\bibitem{TamTsyn}
D.~Tamarkin and B.~Tsygan, ``The ring of differential operators on forms in
  noncommutative calculus,'' in {\em Graphs and Patterns in Mathematics and
  Theoretical Physics}, M.~Lyubich and L.~Takhtajan, eds., Proceedings of
  Symposia in Pure Mathematics, pp.~105--138.
\newblock American Mathematical Society, 2005.

\bibitem{dolgushev2011noncommutative}
V.~Dolgushev, D.~Tamarkin, and B.~L. Tsygan, ``{Noncommutative calculus and the
  Gauss--Manin connection},'' in {\em Higher structures in geometry and
  physics}, pp.~139--158.
\newblock Springer, 2011.

\bibitem{Gerst}
M.~Gerstenhaber, ``{The Cohomology Structure of an Associative Ring},'' {\em
  Ann. Math.} {\bfseries 78} (1963) 59--73.

\bibitem{kadeishvili2008cohomology}
{Tornike Kadeishvili}, ``{Cohomology $C_{\infty}$-algebra and Rational Homotopy
  Type},'' \href{http://arxiv.org/abs/0811.1655}{{\ttfamily arXiv:0811.1655
  [math.AT]}}.

\bibitem{MARKL1992141}
M.~Markl, ``{A cohomology theory for A(m)-algebras and applications},''
  \href{http://dx.doi.org/https://doi.org/10.1016/0022-4049(92)90160-H}{{\em
  Journal of Pure and Applied Algebra} {\bfseries 83} no.~2, (1992) 141--175}.

\bibitem{Lada1994StronglyHL}
T.~Lada and M.~Markl, ``{Strongly homotopy Lie algebras},'' {\em Commun. in
  Algebra} {\bfseries 23} (1994) 2147--2161.

\bibitem{Lada_commutators}
{T. Lada}, ``{Commutators of $A_\infty$ structures},'' {\em Contemp. Math.}
  {\bfseries 227} (1999) 227--233.

\bibitem{GV}
{Gerstenhaber, M., Voronov, A.}, ``{Higher operations on Hochschild complex},''
  {\em Funct. Anal. Appl.} {\bfseries 29} (1995) 1--6.

\bibitem{Cuntz}
J.~Cuntz, G.~Skandalis, and B.~Tsygan, {\em {Cyclic Homology in Non-Commutative
  Geometry}}.
\newblock Springer, 2004.

\bibitem{DomKow}
D.~Fiorenza and N.~Kowalzig, ``{Highe Brackets on Cyclic and Negative Cyclic
  (Co)Homology},'' {\em International Mathematics Research Notices} {\bfseries
  2020} no.~23, (10, 2018) 9148--9209.

\bibitem{Kadeishvili98}
T.~Kadeishvili, ``{The structure of the A($\infty$)-algebra, and the Hochschild
  and Harrison cohomologies.},'' {\em Mat. Inst. Razmadze Akad. Nauk Gruzin.
  SSR} (1998) .

\bibitem{StasheffLada}
J.~Stasheff and T.~Lada, ``{Introduction to SH Lie algebras for physicists},''
  {\em International Journal of Theoretical Physics} {\bfseries 32} no.~7,
  (1993) 1087 -- 1103.

\bibitem{Kadeishvili1980ONTH}
T.~Kadeishvili, ``On the homology theory of fibre spaces,'' {\em Russian
  Mathematical Surveys} {\bfseries 35} (1980) 231--238.

\bibitem{Getzler1990A}
E.~Getzler and J.~D.~S. Jones, ``{$A_\infty$-algebras and the cyclic bar
  complex},'' {\em Illinois Journal of Mathematics} {\bfseries 34} (1990)
  256--283.

\bibitem{Loday}
J.-L. Loday, {\em {Cyclic Homology}}.
\newblock Springer, 1998.

\bibitem{Reinhold}
B.~Reinhold, ``{$L_\infty$-agebras and their cohomology},''
  \href{http://dx.doi.org/10.1051/emsci/2019003}{{\em Emergent Scientist}
  {\bfseries 3} (2019) 4}.


\bibitem{TSUJISHITA19913} T. Tsujishita, ``Homological method of computing invariants of systems of differential equations,'' {\em Differential Geometry and its Applications} {\bfseries 1} (1991) 3 -- 34.


\bibitem{Bryant1995}
Robert L. Bryant and Phillip A. Griffiths,``Characteristic Cohomology of Differential Systems (I): General Theory,''
 {\em Journal of the American Mathematical Society} {\bfseries 8} (1995) 507--596.
 
 
\bibitem{Barnich:2009jy}
G.~Barnich and M.~Grigoriev, ``{A Poincare lemma for sigma models of AKSZ
  type},'' {\em J. Geom. Phys.} {\bfseries 61} (2011) 663--674.

 
 \bibitem{Anderson} Ian M. Anderson and Charles G. Torre, ``Asymptotic Conservation Laws in Classical Field Theory,'' {\em Phys. Rev. Lett.} {\bfseries 77} (1996) 4109--4113.
 


\bibitem{Sharapov:2016sgx}
A. A. Sharapov, ``Variational Tricomplex, Global Symmetries and Conservation Laws of Gauge Systems,'' {\em SIGMA} {\bfseries 12} (2016) 098.


\bibitem{hood1986}
C.~Hood and J.~Jones, ``Some algebraic properties of cyclic homology groups,''
  {\em K-Theory} {\bfseries 1} no.~4, (1987) .

\bibitem{Volkov}
K.~Cieliebak and E.~Volkov, ``{Eight flavors of cyclic homology},''
  \href{http://dx.doi.org/10.1215/21562261-2021-0008}{{\em Kyoto Journal of
  Mathematics} {\bfseries 61} no.~2, (2021) 495 -- 541}.

\end{thebibliography}
\end{document}